\documentclass[12pt]{article}
\usepackage{amsmath}
\usepackage{graphicx}
\usepackage{enumerate}
\usepackage{natbib}
\usepackage{booktabs}
\usepackage{algorithm}
\usepackage{xcolor}
\usepackage{algorithmic}
\usepackage{amsfonts}
\usepackage{amsthm}
\usepackage{amssymb}
\usepackage{url} 

\newcommand{\blind}{1}

\newtheorem{definition}{Definition}
\newtheorem{condition}{Condition}
\newtheorem{lemma}{Lemma}
\newtheorem{proposition}{Proposition}
\newtheorem{theorem}{Theorem}
\newtheorem{corollary}{Corollary}

\newcommand{\bx}{\mathbf{x}}
\newcommand{\bu}{\mathbf{u}}
\newcommand{\bX}{\mathbf{X}}
\newcommand{\bY}{\mathbf{Y}}
\newcommand{\bC}{\mathbf{C}}

\newcommand{\bE}{\mathbf{E}}
\newcommand{\bA}{\mathbf{A}}
\newcommand{\bG}{\mathbf{G}}
\newcommand{\bH}{\mathbf{H}}
\newcommand{\bB}{\mathbf{B}}

\newcommand{\bI}{\mathbf{I}}

\newcommand{\bQ}{\mathbf{Q}}
\newcommand{\bW}{\mathbf{W}}
\newcommand{\br}{\mathbf{r}}

\newcommand{\be}{\mathbf{e}}
\newcommand{\bZ}{\mathbf{Z}}

\newcommand{\bR}{\mathbf{R}}
\newcommand{\bS}{\mathbf{S}}

\newcommand{\bV}{\mathbf{V}}
\newcommand{\tr}{\tilde{r}}

\newcommand{\bw}{\mathbf{w}}
\newcommand{\bbE}{\mathbb{E}}

\newcommand{\bbeta}{\boldsymbol{\beta}}
\newcommand{\bomega}{\boldsymbol{\omega}}
\newcommand{\bOmega}{\boldsymbol{\Omega}}
\newcommand{\bSigma}{\boldsymbol{\Sigma}}

\newcommand{\bDelta}{\boldsymbol{\Delta}}
\newcommand{\bPhi}{\boldsymbol{\Phi}}
\newcommand{\bxi}{\boldsymbol{\xi}}
\newcommand{\bvarepsilon}{\boldsymbol{\varepsilon}}
\newcommand{\calY}{\mathcal{Y}}
\newcommand{\calV}{\mathcal{V}}
\newcommand{\calE}{\mathcal{E}}
\newcommand{\calA}{\mathcal{A}}
\newcommand{\calM}{\mathcal{M}}
\newcommand{\calC}{\mathcal{C}}
\newcommand{\calR}{\mathcal{R}}
\newcommand{\hbu}{\hat{\mathbf{u}}}

\begin{document}

\def\spacingset#1{\renewcommand{\baselinestretch}%
{#1}\small\normalsize} \spacingset{1}


\if1\blind
{
  \title{\bf MultiFun-DAG: Multivariate Functional Directed Acyclic Graph}
  \author{
  Tian Lan \\
    Department of Industrial Engineering, Tsinghua University\\
    Ziyue Li\\
    Information Systems Department, University of Cologne \\
    Junpeng Lin\\
    Department of Industrial Engineering, Tsinghua University\\
    Zhishuai Li\\
    Sensetime \\
    Lei Bai \\
    Shanghai AI Laboratory \\
    Man Li \\
    Department of Industrial Engineering and Decision Analytics, \\ The Hong Kong University of Science and Technology \\
    Fugee Tsung \\
    Department of Industrial Engineering and Decision Analytics, \\ The Hong Kong University of Science and Technology \\
    Rui Zhao \\
    Sensetime \\
    Chen Zhang \\
    Department of Industrial Engineering, Tsinghua University
}
  \maketitle
} \fi

\if0\blind
{
  \begin{center}
    {\LARGE\bf  MultiFun-DAG: Multivariate Functional Directed Acyclic Graph}
\end{center}
} \fi

\begin{abstract}
Directed Acyclic Graphical (DAG) models efficiently formulate causal relationships in complex systems. Traditional DAGs assume nodes to be scalar variables, characterizing complex systems under a facile and oversimplified form. This paper considers that nodes can be multivariate functional data and thus proposes a multivariate functional DAG (MultiFun-DAG). It constructs a hidden bilinear multivariate function-to-function regression to describe the causal relationships between different nodes. Then an Expectation-Maximum algorithm is used to learn the graph structure as a score-based algorithm with acyclic constraints. Theoretical properties are diligently derived. Prudent numerical studies and a case study from urban traffic congestion analysis are conducted to show MultiFun-DAG's effectiveness. 
\end{abstract}

\noindent%
\vfill

\newpage
\spacingset{1.9} 
\section{Introduction}
Directed acyclic graph (DAG), a.k.a., Bayesian network, is a probabilistic graphical model that represents a set of variables and their causal relationships. In a DAG, each node corresponds to a random variable, and each directed acyclic edge represents a causal dependence relationship between the two variables, i.e., a parent node and a descendant node. The distribution of each variable can be written as a conditional probability distribution given its parent nodes and is independent from other nodes. DAG has been widely used to offer vital insights for causal relationship discovery in biological \citep{aguilera2011bayesian}, physical \citep{velikova2014exploiting}, social systems \citep{ruz2020sentiment}, etc. 

Previous work has thoroughly studied DAG with each node as a scalar variable \citep{heckerman2008tutorial}. However, it is common to come across systems where the variables have a \textit{functional} form, as shown in Fig. \ref{fig:intro}. (b). Functional data is formally defined as the data with each sample in the form of random curves or functions over a continuum, such as time or space \citep{qiao2019functional}, which is commonly observed in complex systems such as medical science \citep{chen2018functional}, physiology \citep{li2018nonparametric}, and climate \citep{fraiman2014detecting}. For example, in urban transportation, sensors collect the real-time signals of the traffic elements, such as traffic volume, vehicle speed, lane saturation, cycle length of traffic lights, and weather, which are all functional data and can be combined into a multivariate form. By modeling these traffic variables as different nodes in a DAG to learn their causal relationships, root causes for traffic congestion can be identified, and then corresponding actions can be taken \citep{10.1145/3580305.3599436}. 

We consider the DAG in which each node can be multivariate functional data, as it can describe the practical systems more pertinently than the scalar-based ones. 
Such a DAG has three critical properties: (1) \textbf{Infinite dimensionality}: Functional data are naturally infinite-dimensional, and in theory, can have infinitely many points; Though in practice functional data is usually discretized or approximated to a finite number of observation points. However, the theoretical foundation is that the true underlying functional observation is of infinite dimensionality. 
(2) \textbf{Data heterogeneity}: Functions of different nodes can be heterogeneous, such as containing various numbers of functions and coming from different spaces. 
(3) \textbf{Inter-causation}: Functions of different nodes could be inter-correlated in different ways, i.e., different functions of one node can have different causal effects on another function of another node. 

\begin{figure}[t]
    \centering
    \includegraphics[width=0.8\linewidth]{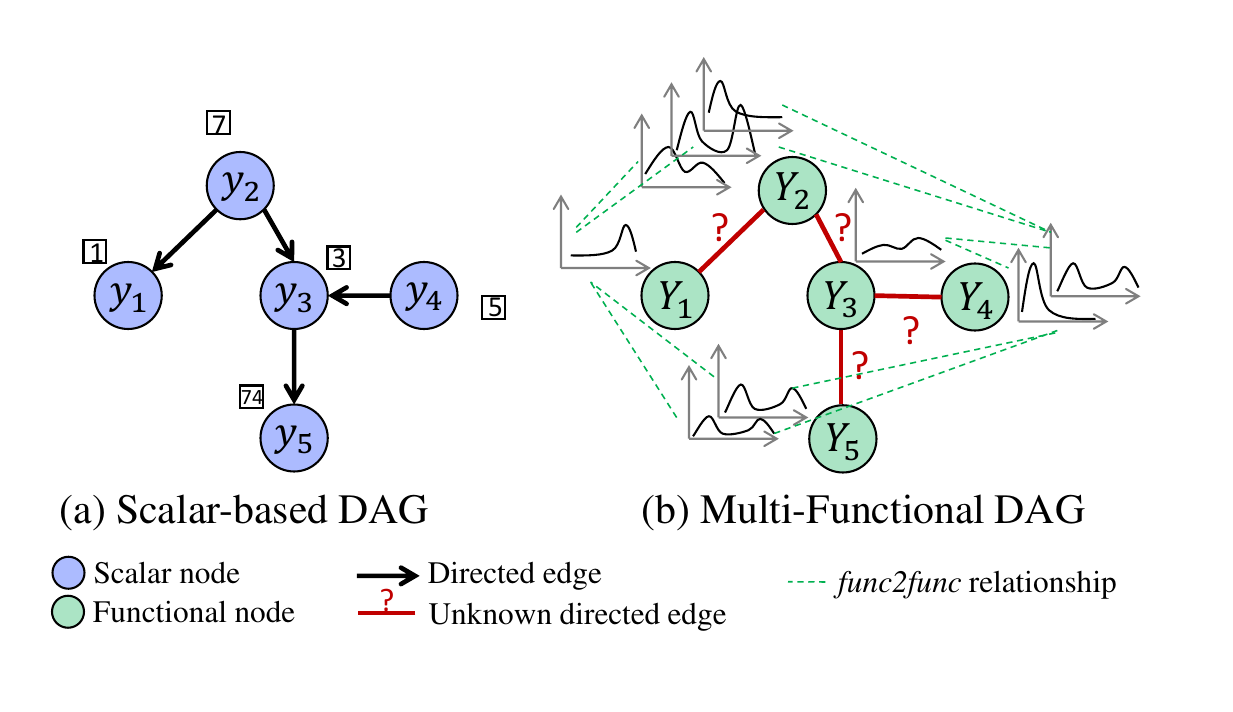}
    \caption{Scalar-based DAG v.s.  Multi-Functional DAG. Each node is a scalar or functional variable, and the directed edge is the causal dependence. MultiFun-DAG learns the unknown causal edge (solid) via formulating the \textit{func2func} relationships (dotted). 
    }
    \label{fig:intro}
\end{figure} 

As a result, traditional scalar-based DAGs cannot be easily extended to our case. 

This paper aims to build a \textbf{Multi-Fun}ctional \textbf{DAG} (MultiFun-DAG) to learn the valuable causal dependence structure among different multi-functional nodes. 
The task is unfolded by three concrete questions: (1) how to preserve the information and describe causal dependence relationships for infinite functions? (2) how to model and fuse the causal dependence relationships between multiple functions in any two nodes and build an edge between them? (3) how to conduct structural learning and parameter learning for these edges?

To address these challenges, we are the first to propose a novel DAG to learn the causal structure with multivariate functional data, with the following major contributions:
\begin{itemize}
\setlength\itemsep{0.5pt}
    \item We model the causal dependence relationships between nodes with multiple functions via hidden bilinear \textit{function-to-function (\textit{func2func}) regression}  with \textit{low-rank decomposition}. 
    \item We propose an Expectation-Maximization (EM) algorithm in the score-based structural learning framework to learn the DAG structure with acyclic constraint and group lasso penalty. 
    \item We derive the theoretical properties of the model, including its identifiability and asymptotic error bound of the EM algorithm, and the asymptotic oracle property of our structure learning algorithm. 
\end{itemize}

\section{Related Work}
\label{sec:literature}

\subsection{DAG structural learning methods}
Methods for DAG learning can be categorized into combinatorial learning and continuous learning algorithms.

\textbf{Combinatorial learning algorithms} solve a combinatorial optimization problem to find whether an edge exists between any two nodes. This type of method can be further divided into constraint-based and score-based algorithms. 

Constraint-based methods, such as PC \citep{spirtes2000causation}, rankPC \citep{harris2013pc}, and fast causal inference \citep{spirtes2000causation}, learn the edges by conditional independence tests. However, 
they are built upon that the independence tests should accurately reflect the independence model, which is generally difficult to be satisfied in reality. As a result, these methods suffer from error propagation, where a minor error in the early phase can result in a very different DAG.

The score-based methods instead construct a score function to evaluate DAG structures and select the graph with the highest score. Some commonly used score functions include the likelihood function, mean square fitting error, etc. Some further regularization items on edges are also added in the score to learn a sparse graph  \citep{chickering2002optimal,nandy2018high}.
Then greedy searches are implemented to find the graph with the highest score. However, one drawback of the combinatorial score-based method is the nonconvexity of the combinatorial problem. The acyclicity constraint means that the solution space stretches along all topological orderings that have $d!$ permutations in a graph with $d$ nodes, rendering DAG learning an NP-hard problem. 

\textbf{Continuous learning algorithms} formulate the acyclic constraint into an algebraic form and convert the structure learning problem into a purely continuous optimization problem to save computation cost. In particular, \citet{zheng2018dags} proposed NoTears, which formulates an algebraic form as $h(W) = \text{tr}(\exp(W \circ W)) - d = 0$, where $W$ is the adjacency weight matrix, $\text{tr}(\cdot)$ is the trace, and $\circ$ is Hadamard product. 
This idea was popularly borrowed in many preceding works. For example, \citet{zheng2020learning} develops a nonparametric DAG based on NoTears \citet{bhattacharya2021differentiable} 
considers both directed and undirected edges based on NoTears. Besides, \citet{ng2020role} also proposes a soft constraint for acyclicity.
However, the NoTears-based methods only offer solutions for scalar-variable nodes. The more realistic problem where nodes contain heterogeneous multi-functional data has never been addressed so far. 

\subsection{Functional graphical models}
Functional graphical models (FGMs), as an extension of traditional graphical models, describe the probabilistic dependence between nodes with functional data and could potentially offer solutions for functional DAG learning. According to the direction of the edges, FGMs can be divided into undirected FGMs and directed FGMs. 

The undirected FGMs focus on estimating the correlation dependence structure between different nodes. In particular, \citet{qiao2019functional} proposes a functional graphical Lasso model to describe the sparse correlation dependence structure of different functional nodes. As an extension, \citet{qiao2020doubly} proposes a doubly FGM to capture the evolving conditional dependence among functions. Later more FGMs were proposed,
such as using nonparametric additive conditional independence model \citep{li2018nonparametric}, assuming the dependence to be partially separable \citep{zapata2022partial}, or heterogeneous \citep{wu2022monitoring}, etc. However, undirected FGMs only capture the correlations, instead of causation, of nodes.  

For directed FGMs focusing on the causal relationship of nodes, the current research is scarce. \citet{sun2017functional} proposes a DAG that considers both scalar and functional nodes. Yet it assumes the DAG structure is known in advance. \citet{gomez2020functional} considers DAG with each node as a univariate function. However, it still assumes the topological ordering of nodes should be known in advance by domain knowledge, and transforms the structural learning problem into a parameter selection problem, i.e., selecting the parent node from the candidate parent set. Furthermore, \citet{gomez2020functional} is a two-step framework by first adopting functional principal component analysis (FPCA) to extract features for each node separately, and then using the FPCA scores to model the causal effects. However, since its FPCA totally ignores the causal relationships between different nodes, the extracted PCs may not represent the most useful information in the whole network. Then the causal effects estimated based on these PCs may be misleading and lead to higher estimation errors. 


\section{Proposed Model}
\label{sec:method}
Suppose that a graph $\mathcal{G}=(\calV, \calE)$ represents a DAG with a vertex set $\calV \in \mathbb{R}^P$ and an edge set $\calE \in \mathbb{R}^{P \times P}$, with $P$ denoted as the total number of nodes. A tuple $(j,j') \in \calE$ represents a directed edge leading from node $j$ to node $j'$, i.e., $j \to j'$. Here we assume the node $j$ has $L_j$ functional variables, with
$Y_{jl}(t), t\in \Gamma$ denoted as its $l$-th function, for $l=1,2,...,L_j$. Here without loss of generality, we assume $\Gamma=[0,1]$ is a compact time interval. Suppose we have $N$ identically and independently distributed samples. The $n$-th sample, $n=1,\ldots,N$, is formulated as $\textbf{Y}^{(n)}(t)=(\textbf{Y}^{(n)}_1(t), \textbf{Y}^{(n)}_2(t),...,\textbf{Y}^{(n)}_P(t))^T$, with $\textbf{Y}^{(n)}_j(t) = (Y^{(n)}_{j1}(t), Y^{(n)}_{j2}(t),...,Y^{(n)}_{jL_j}(t))$. Therefore, $\textbf{Y}^{(n)}(t)$ represents $L=\sum_j L_j$ functions of all the nodes, which is a vector. 
Our MultiFun-DAG aims to learn the causal relations between different nodes, i.e., the edge set $\calE$, shown as red lines in Fig. \ref{fig:intro}. 

To achieve it, in Section \ref{sec: method-known}, we first assume that the causal structure $\calE$ is known, and construct a hidden bilinear \textit{func2func} regression, to learn the conditional dependence from the function $l$ of node $j$ to the function $l'$ of node $j'$, shown as the green dotted edge in Fig. \ref{fig:intro}. In Section \ref{sec:non-identifiability}, we show that causal structure is non-identifiable under maximum likelihood estimation. Therefore, we introduce a restriction for DAG structure and its necessity. In Section \ref{sec:-EM}, we combine the restriction in Section \ref{sec:non-identifiability} and propose an EM algorithm for learning the causal structure of MultiFun-DAG.
\subsection{Multi-functional DAG with known structure}
\label{sec: method-known}


We first give an overview of our MultiFun-DAG in Fig. \ref{fig:Methodology}. Our function $Y_{jl}(t)$ follows Gaussian distribution in Eq. (\ref{equ:Y-likelihood}) with mean function $\mu_{jl}(t)$. The mean functions follow \textit{func2func} regression in Eq. (\ref{equ:mean-transition}) with their parents in DAG. To preserve the information for infinite functional variables
, we decompose the mean function into a basis set with coefficients in Eq. ({\ref{equ:FPCA}}). Then we conduct a bilinear regression for the coefficients to describe the linear causality of different nodes via Eq. ({\ref{equ:matrix-bilinear-trans}}). The joint likelihood of coefficients of all the nodes can be represented using a linear Structural Equation Model (SEM) (Eq. \eqref{equ:linear-SEM}).

\begin{figure}[t]
    \centerline{\includegraphics[width=0.8\columnwidth]{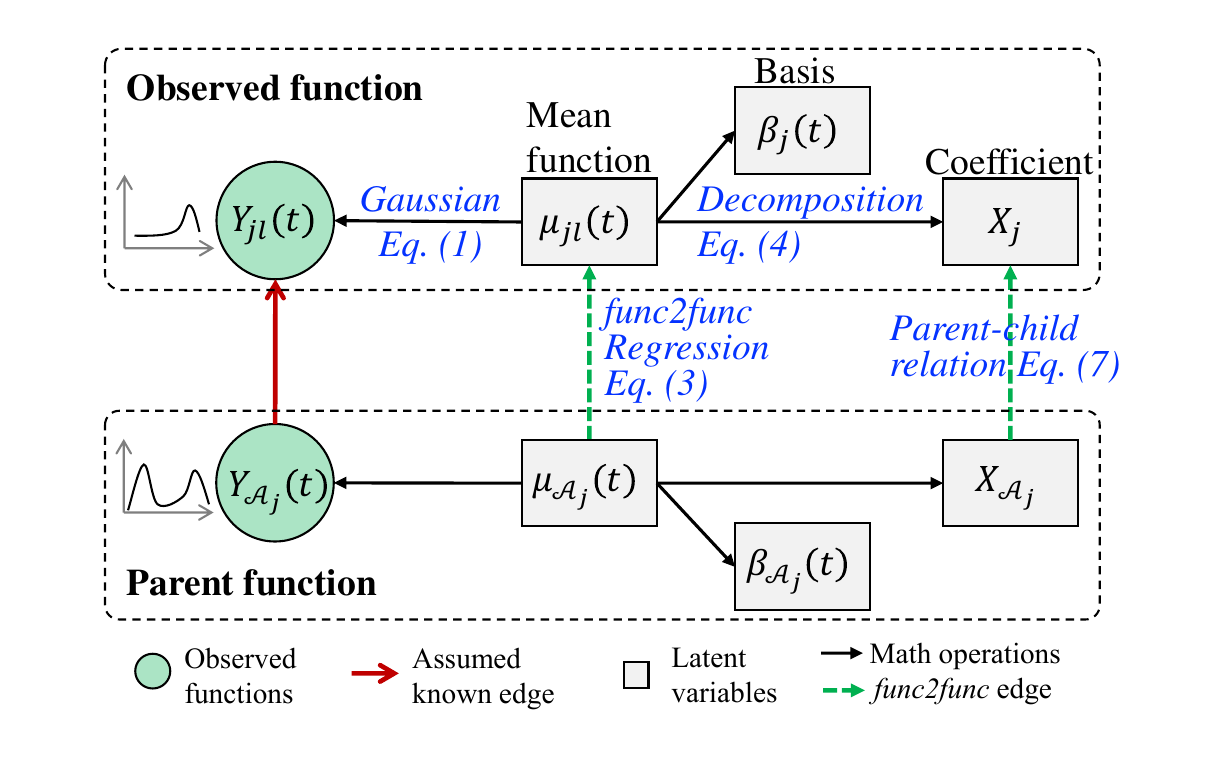}}
    \caption{The Illustration of MultiFun-DAG}
    \label{fig:Methodology}
\end{figure}

In this paper, we focus on Gaussian distributed function: 
\begin{equation} \label{equ:Y-likelihood}
     Y_{jl}^{(n)}(t) \sim \mathcal{N}(\mu_{jl}^{(n)}(t), R_{jl}(\cdot, \cdot)),
\end{equation}
where $\mu_{jl}^{(n)}(t)$ is the mean function and $R_{jl}(\cdot,\cdot)$ is the covariance function of $Y_{jl}^{(n)}$. We assume that $R_{jl}(t,t') = r_{jl}^2 \mathbb{I}(t=t')$, where $r_{jl}^2$ is the scale of variance.  
The parent set of node $j$ is denoted as $\calA_j = \{j'|j'\in \calV, j' \neq j, (j',j)\in \calE\}$. We assume that the joint distribution of $\mu_{jl}^{(n)}(t)$ of all the nodes can be written as the production of the conditional distribution of each node, i.e.,
\begin{equation}
    p(\mu_{11}^{(n)}(t),\ldots, \mu_{PL_{P}}^{(n)}(t)) = \prod_{j=1}^P \prod_{l=1}^{L_{j}}p(\mu_{jl}^{(n)}(t)|\calA_j).
\end{equation}
We focus on linear conditional dependence relationship for $p(\mu^{(n)}_{jl}(t)|\calA_j)$, which is formulated as below:
\begin{equation}
\label{equ:mean-transition}
    \mu_{jl}^{(n)}(t) = \sum_{j'\in \calA_p}\sum_{l'=1}^{L_{j'}} \int_0^1 \gamma_{j'jl'l}(t,s) \mu_{j'l'}^{(n)}(s){\rm d}s + \varepsilon_{jl}^{(n)}(t),
\end{equation}
where $\varepsilon_{jl}^{(n)}(t)$ is the noise function. $\gamma_{j'jl'l}(t,s)$ is the coefficient function for $(j', j)\in \calE$, $l=1,2,...,L_j$ and $l'=1,2,...,L_{j'}$, which describes the contribution of the $l'$-th function of node $j'$ to the $l$-th function of node $j$. We represent $\gamma_{j'jl'l}(t,s)$ and $\mu_{jl}(t)$ as follows:

\textbf{For $\mu_{jl}(t)$:} Given they are in infinite dimensions and hard to be estimated directly, it is common to decompose them into a well-defined continuous space for feature extraction:
\begin{equation} \label{equ:FPCA}
    \mu^{(n)}_{jl}(t) = \sum_{k=1}^{K_j} x^{(n)}_{jlk} \beta_{jk}(t),
\end{equation}
where $\bB_{j}(t)=(\beta_{j1}(t),\beta_{j2}(t),...,\beta_{jK_{j}}(t))^T$ is an orthonormal functional basis set for node $j$, with $\int\beta_{jk}(t)^{2} {\rm d}t=1, k=1,
\ldots, K_{P}$ and $\int \beta_{jk}(t)\beta_{jk'}(t) {\rm d}t = 0$, $k\neq k'$. $x_{jlk}^{(n)}$ is the corresponding coefficient. 

\textbf{For $\gamma_{j'jl'l}(t,s)$:} , we describe $ \gamma_{j'jl'l}(t,s)$ using the corresponding basis sets in a bilinear way \citep{hoff2015multilinear} as:
\begin{align} \label{equ:bilinear}
    \gamma_{j'jl'l}(t,s) = \sum_{k=1}^{K_j} \sum_{k'=1}^{K_{j'}} c_{j'jk'k} \cdot c_{j'jl'l}\beta_{j'k'}(s) \beta_{jk}(t).
\end{align}
$c_{j'jk'k}$ represents the influence caused by the basis pair: $\beta_{j'k'}(s)$ on $ \beta_{jk}(t)$. $c_{j'jl'l}$ represents the influence caused by the function pair: function $l'$ of node $j'$ on function $l$ of node $j$. This decomposition describes the regression coefficient function from two aspects, i.e., (1) the basis set of a node and (2) the variables of a node, separately. Besides, it also improves estimation stability. 

By plugging the representation of Eq. (\ref{equ:FPCA}) and (\ref{equ:bilinear}) into Eq. (\ref{equ:mean-transition}), for function $l$ in node $j$, we could obtain:
\begin{equation} \label{equ:bilinear-trans}
\begin{split}
    \sum_{k=1}^{K_j} x^{(n)}_{jlk} \beta_{jk}(t)= \sum_{j'\in \calA_j} \sum_{k=1}^{K_j} \sum_{l'=1}^{L_{j'}} \sum_{k'=1}^{K_{j'}} \int_0^1 c_{j'jk'k} \cdot c_{j'jl'l} x^{(n)}_{j'l'k'} 
   \beta_{jk}(t)  \beta_{j'k'}^2(s){\rm d}s  + \varepsilon_{jl}(t).
\end{split}
\end{equation}
By integrating this equation over $s$, and combining all the parameters $x^{(n)}_{jlk}$ into a vector, i.e., $\bx^{(n)}_j = {\rm vec}(x^{(n)}_{jlk}) \in \mathbb{R}^{L_jK_j}$, where $[\bx^{(n)}_{j}]_{i}$ represents the $[(i-1) \text{ mod } K_j]+1$ coefficient of the function $\lfloor (i-1)/K_j \rfloor +1$ in node $j$, 
Eq. (\ref{equ:bilinear-trans}) can be re-written as:
\begin{align} 
\label{equ:matrix-bilinear-trans}
    \bx^{(n)}_j = \sum_{j' \in \calA_j} (\bC^L_{j'j} \otimes \bC^K_{j'j})^T \bx^{(n)}_{j'} + \bxi_j^{(n)}.
\end{align}
Here $\bC^L_{j'j}\in \mathbb{R}^{L_{j'} \times L_j}$ with $[\bC^L_{j'j}]_{l'l}=c_{j'jl'l}, \bC^K_{j'j}\in \mathbb{R}^{K_{j'} \times K_{j}}$ with $[\bC^K_{j'j}]_{k'k}=c_{j'jk'k}$. $\otimes$ is the Kronecker product. $\bxi_j \in \mathbb{R}^{L_jK_j} $ is the noise of $\bx_j$, where $[\bxi_j]_{(l-1)K_j+1}$ to $[\bxi_j]_{lK_j}$ are the projection of $\varepsilon^{(n)}_{jl}(t)$ on its corresponding basis set for $j=1,\ldots,P, l=1,\ldots,L_j$. Here we assume $\bxi_j^{(n)}\sim \mathcal{N}(\mathbf{0},\bOmega _j )$ with $\bOmega _j \in \mathbb{R}^{L_j K_j \times L_j K_j}$. For brevity, we simply assume $\bOmega_j$ has a diagonal form, i.e., $\bOmega_j = {\rm diag}(\bomega_j^2)$.

Lastly, we use a linear SEM to interpret our MultiFun-DAG. We denote $\mathbf{C}\in \mathbb{R}^{M\times M}$ with its $(j,j')$ block as $\bC_{(j',j)}=\bC_{j'j}$, $\bC_{j'j}=\bC^L_{j'j} \otimes \bC^K_{j'j}$ if $(j,j') \in \calE$, otherwise we have $\bC_{(j',j)} = \mathbf{0}_{L_{j'}K_{j'} \times L_jK_j}$. Then for $\textbf{x}^{(n)}=[\bx_1^{(n)},\ldots,\bx_{P}^{(n)}] \in \mathbb{R}^M$, a linear SEM interpretation is:
\begin{align} \label{equ:linear-SEM}
    \mathbf{x}^{(n)} = \mathbf{C}^T \mathbf{x}^{(n)} + \bxi^{(n)}.
\end{align}
Here $M=\sum_{j=1}^P L_jK_j$, $\bxi^{(n)} = [\bxi^{(n)}_{1},\ldots,\bxi^{(n)}_{P}] \sim \mathcal{N}(\mathbf{0},\bOmega)$ is the noise vector. $\bOmega= {\rm diag}(\bOmega_{1},\ldots,\bOmega_{P})$. 

In reality, $Y_{jl}(t)$ can only be measured at certain discrete observation points. In this paper, without loss of generality, we assume that, for all the nodes, the sampling points are equally spaced as $t_1,\ldots,t_{T}$. Then we define $\bX=[\bx^{(1) T},\ldots,\bx^{(N) T}]^T \in \mathbb{R}^{N \times M}, \mathbf{Y}_{jl}^{(n)}=[Y_{jl}^{(n)}(t_1),\ldots, Y_{jl}^{(n)}(t_T)]$. By abusing the notation $\mathbf{Y}^{(n)}=[\mathbf{Y}_{jl}^{(n)}, j = 1,\ldots,P, l = 1,\ldots, L_j]$ and $\calY=[\mathbf{Y}^{(1)},\ldots, \mathbf{Y}^{(N)}]$ 
for convenience, we can write the joint likelihood of the generative model as:
\begin{align} \label{equ:joint-likelihood}
    f(\bX,\calY) = \prod_{i=1}^{N} p(\bx^{(n)}) p(\mathbf{Y}^{(n)}|\bx^{(n)}),
\end{align}
where $p(\bx^{(n)}) $ and $p(\mathbf{Y}^{(n)}|\bx^{(n)})$ are computed by Eqs. (\ref{equ:Y-likelihood}), (\ref{equ:FPCA}) and  (\ref{equ:linear-SEM}). It is to be noted that our model can also be applicable to functional nodes measured at distinct observation points with different lengths, with trivial notation modifications. 

\subsection{Non-identifiability and  equivalence class} \label{sec:non-identifiability}
In reality, the graph structure is unknown and to be estimated. This can be transferred to infer whether the weight $\bC_{j'j}^{L}$ and $\bC_{j'j}^{K}$ equals $\mathbf{0}$ for certain blocks. In particular, the parameters to be estimated in our model includes 1) the weights $\bC_{j'j}^{L}$ and $\bC_{j'j}^{K}$ for nodes $j,j'=1,\ldots,P$; 2) the variance of functional noise, denoted as $\br=[r_{11}^2, r_{12}^2, ..., r_{PL_P}^2]\in \mathbb{R}^{M}$; 3) the variance of $\bx_{j}$, i.e., $\bOmega_1,\bOmega_2,...,\bOmega_P$, denoted as $\bOmega_{[1:P]}$; 4) the basis functions $\bB_{j}(t)=[\beta_{j1}(t),\beta_{j2}(t),...,\beta_{jK_{j}}(t)]^T$ for node $j=1,\ldots, P$. 
    It is to be noted that in reality, we only need to estimate $\bB_{j}=[\bB_{j}(t_1)^{T},\ldots,\bB_{j}(t_{T})^{T}]^{T}$, denoted as $\bB=[\bB_1, \bB_2, \ldots, \bB_P]$. For other observation points, we can adopt Kernel smoothing to estimate them easily. 
    
     The parameters $\Theta = (\bC, \bB, \br, \bOmega_{[1:P]})$ are statistically nonidentifiable without further constraints. Based on our model structure, the marginal distribution of $\bY$ follows a Gaussian distribution with mean $\mathbf{0}$ and covariance function $\bSigma_{\bY}(\Theta)$, which is determined by the parameters $\Theta$, i.e.,
    \begin{equation} \label{equ:equivalence-class-con}
        \bSigma_{\bY}(\Theta)_{jl,j'l'} = 
        \begin{cases}
            \bB_j[(\bI-\bC)^{-T}\bOmega(\bI-\bC)^{-1}]_{jl,jl} \bB_j^T + r_{jl}^2 \mathbf{I}_T & (j,l)=(j',l') \\
            \bB_j[(\bI-\bC)^{-T}\bOmega(\bI-\bC)^{-1}]_{jl,j'l'} \bB_{j'}^T & o.w.
        \end{cases}.
    \end{equation}

    We aim to estimate the model parameters $\Theta$ based on the information from the observed covariance matrix $\bSigma_\bY$. However, it turns out that the mapping from $\Theta$ to $\bSigma_\bY(\Theta)$ is not one-to-one, i.e., one $\bSigma_\bY$ can correspond to multiple sets of model parameters $\Theta$. Denote the true covariance matrix as $\bSigma_\bY^*=\bSigma_\bY(\Theta^*)$, where $\Theta^*$ is the true underlying parameters. We define the set of all $\Theta$ whose $\bSigma_\bY(\Theta)$ equals $\bSigma_\bY^*$ as the equivalence class $\mathfrak{D}$ corresponding to $\bSigma_\bY^*$, i.e.,
    \begin{equation*}
        \mathfrak{D} = \{\Theta | \bSigma_\bY^* = \bSigma_\bY(\Theta)\}.
    \end{equation*}
    
    Without additional restrictions, we can only find one $\Theta \in \mathfrak{D}$ based on the observation data. However, infinite combinations of parameters exist in the equivalence class and cannot provide us with useful information regarding the causal structure. The most common solution for Gaussian noise is to assume Condition \ref{ass:equal-variance}, which can be viewed as an extension of the equal variance condition in \citet{van2013ell_}. 
    
    \begin{condition} \label{ass:equal-variance}
        In the true DAG, all latent variables have equal variance, i.e., $\bOmega = \omega_0^2 \bI$ 
    \end{condition}
It is a common condition for ensuring the identifiability of a linear structural causal model with Gaussian noise. With it, all the graphs with $\Theta \in \mathfrak{D}$ will have the same causal structure.

\subsection{Regularized EM estimation} \label{sec:-EM}

Under Condition \ref{ass:equal-variance}, we rewrite the parameter set as $\Theta=\{\bC,\bB,\br,\omega_0^2\}$. 
Since the coefficients $\bX$ are unknown, we estimate $\bx^{(n)}, n=1,\ldots, N$ by treating them as latent variables and using a regularized EM algorithm for estimation \citep{yi2015regularized}. The regularized EM algorithm consists of an Expectation-step and a regularized Maximization-step. In each iteration, the operator $\calM_n$ of the regularized EM is denoted as follows: 
\begin{equation}
\label{eq:MM}
\begin{split}
    \calM_n(\Theta') &= \underset{\Theta}{\arg \max} \ Q_n(\Theta;\Theta') - \lambda \calR(\bC) \\ 
    s.t. & \quad \quad \mathcal{G} \text{ is a DAG},
\end{split}
\end{equation}
where
    \begin{equation}
        \begin{split}
        Q_n(\Theta;\Theta') & = \mathbb{E}_{\bX|\calY;\Theta'}\log f(\bX,\calY;\Theta) = \int \log f(\bX,\calY;\Theta') p(\bX|\calY;\Theta') {\rm d}\bX, \\
        \end{split}
    \end{equation}
    \begin{align}
    \label{eq:logf}
       \log f(\bX, \calY;\Theta) = & -\frac{1}{2} \Bigg( \sum_{n=1}^N \Big( \sum_{j=1}^P \sum_{l=1}^{L_j} (\bY_{jl}^{(n)}-\bB_j\bx_{jl}^{(n)})^T r_{jl}^{-2}(\bY_{jl}^{(n)}-\bB_j\bx_{jl}^{(n)}) \nonumber \\
       &+ \sum_{j=1}^P (\bx_j^{(n)} - \bx^{(n)} \bC_j)^T\omega_0^{-2}(\bx^{(n)}_j - \bx^{(n)} \bC_j) + \sum_{j=1}^P \sum_{l=1}^{L_j} T \log r_{jl}^2 \nonumber\\ 
       &+ M \log \omega_0^2 \Big) \Bigg) + constants,
    \end{align}  
and $\calR(\bC)$ is the sparse penalty, to penalize the model complexity. 

 To represent the DAG constraint in Eq. (\ref{eq:MM}) to a mathematical form, we define the adjacency matrix $\bW \in \mathbb{R}^{P \times P}$ corresponding to the edge set $\calE$ for the DAG $\mathcal{G}$.  Consider $\bW$ as a measure of causal effects and it fuses the information of $\bC_{ij}$ in a scalar. We have:
    \begin{align}
        [\bW]_{ij} \ne 0 \Leftrightarrow \bC_{ij} \ne \textbf{0}_{L_iK_i \times L_jK_j}.
    \end{align}
   Then in this work, we give an intuitive and valid definition for $\bW$ as 
   \begin{equation}
   [\bW]_{ij}  \doteq \| \bC_{ij}\|_F.
   \end{equation}
Consequently, to ensure $\bC$ is a DAG, we adopt Notears constraints \citep{zheng2018dags} for the adjacency matrix $\bW$ that $h(\bW) := {\rm tr}(\exp(\bW \circ \bW)) - P$ and we have:
    \begin{equation}
        h(\bW) = 0\Leftrightarrow {\mathcal{G}} \text{ is a DAG}.
    \end{equation}

Finally, for a large graph, it is usually assumed the edges are sparse, and penalize the $l_{1}$ norm of $\mathbf{W}$. Therefore, we set $\calR(\bC) =\|\bC\|_{l_1/F}=\sum_{i=1}^{P}\sum_{j=1}^{P}\|\bC_{ij}\|_{F}$, and $\lambda$ adjusts the strength of the penalty.

\textbf{Expectation-step} is to calculate $Q_{n}(\Theta;\Theta^{\prime})$. It can be derived by calculating the posterior likelihood $p(\bX|\calY;\Theta')$, which can be estimated in a forward and backward way. 
\begin{proposition} \label{pro:represent}
For any parameter set $\Theta'$, the posterior distribution can be decomposed as $p(\bX|\calY;\Theta')=\prod_{i=1}^{N}p(\bx^{(n)}|\bY^{(n)};\Theta')$. $p(\bx^{(n)}|\bY^{(n)};\Theta')$ follows a multivariate normal distribution $ \mathcal{N}(\hbu^{(n)}_{\Theta',\bY^{(n)}} , \hat{\bSigma}_{\Theta'})$ with mean $\hbu^{(n)}_{\Theta',\bY^{(n)}} \in \mathbb{R}^{1\times P}$ and variance $\hat{\bSigma}_{\Theta'} \in \mathbb{R}^{P\times P}$, where $\hbu_{\Theta, \bY}$ is a linear combination of $\bY$ depending on $\Theta$ while $\hat{\bSigma}_{\Theta}$ only depends on $\Theta$.
\end{proposition}
\begin{proof}
    Proposition \ref{pro:represent} is straightforward by following the procedure in Appx. \ref{app:Expectation}. 
\end{proof}

\textbf{Maximization-step} is to solve the maximization problem of Eq. (\ref{eq:MM}) based on the calculated $Q_{n}(\Theta,\Theta')$ in the Expectation-step, and update the model parameters $\Theta$. The parameters can be decoupled into two sub-groups. The first sub-group is $\bB_{j}(t),j=1,\ldots,P$ and $\br$, which are directly related to the observations $\calY$. The second sub-group contains the important $\bC$ and $\omega_0^2$, which determine the causal relationship of different nodes, namely, the DAG structure.

    \textbf{For the first sub-group:}
 denote $F_n(\bB, \Theta')$ as the part of the quadratic loss in Eq. (\ref{eq:MM}) related to $\bB$. Minimizing $F_n$ is equivalent to maximizing $Q_n$ respecting to $\bB$:
    \begin{equation} 
    \label{equ:M-step-u1}
    \begin{split}
    \small
        \hat{\bB}_{1},\hat{\bB}_{2},\ldots,\hat{\bB}_{P} \doteq\underset{\bB_{1},\ldots,\bB_{P}}{\arg \min} & \sum_{n=1}^N \sum_{j=1}^P \sum_{l=1}^{L_j} \bbE_{\bx^{(n)}|\bY^{(n)},\Theta'} \left((\bY_{jl}^{(n)}-\bB_j\bx_{jl}^{(n)})^T r_{jl}^{-2}(\bY_{jl}^{(n)}-\bB_j\bx_{jl}^{(n)})\right) \\
        & \propto \frac{1}{N} \sum_{n=1}^N \sum_{j=1}^P \sum_{l=1}^{L_j} \left(\| \bY^{(n)}_{jl} - \bB_j \hbu_{jl;\Theta', \bY^{(n)}} \|^2_2 + {\rm tr}(\bB_j \hat{\bSigma}_{j} \bB_j^T)\right) \\
        & \doteq F_n(\bB,\Theta') \\ 
        & {\rm s.t.} \quad \bB_j^T\bB_j=\bI \quad \forall j=1,\ldots, P.
        \end{split}
    \end{equation}
To solve Eq. (\ref{equ:M-step-u1}), we can utilize the polar decomposition. We first calculate $\bA = \frac{1}{N} \sum_{n=1}^N \sum_{l=1}^{L_j} \bY_{jl}^{(n)} \hbu_{jl;\Theta',\bY^{(n)}}^T$, and then perform the polar decomposition on $\bA$ to obtain $\bA = \bV \hat{\bB}_j$, where $\bV$ is a symmetric matrix, and $\hat{\bB}_j$ is the  matrix we are interested in.

For estimating $\hat{\br}$, it can be solved in a closed form as: 
    \begin{equation} \label{equ:M-step-u2}
    \small
        \hat{r}_{jl}^2 = \frac{1}{N} \sum_{n=1}^N \left((\bY^{(n)}_{jl} - \hat{\bB}_j \hbu_{jl;\Theta', \bY^{(n)}})(\bY^{(n)}_{jl} - \hat{\bB}_j \hbu_{jl;\Theta', \bY^{(n)}})^T + \hat{\bB}_j\hat{\bSigma}_{j} \hat{\bB}_j^T)\right), \forall j = 1,\ldots,P,  l = 1,\ldots, L_{j}. 
    \end{equation}
   
    \textbf{For the second sub-group:} \textbf{the key is to infer $\bC$, which represents the structure of DAG}. Denote $G_n$ as the loss part in $Q_n(\Theta,\Theta')$ related to $\bC$. Maximizing Eq. (\ref{eq:MM}) is equivalent to the following:
    \begin{equation} \label{equ:M-step-c1}
    \begin{split}
       \hat{\bC}^K, \hat{\bC}^L =   \underset{\bC^K, \bC^L}{\arg\min} &  \sum_{n=1}^N \sum_{j=1}^P \bbE_{\bx^{(n)}|\bY^{(n)},\Theta'} \left((\bx_j^{(n)} - \bx^{(n)} \bC_j)^T\omega_0^{-2}(\bx^{(n)}_j - \bx^{(n)} \bC_j)\right) + \lambda \| \bC\|_{l_1/F} \\
        \propto & \frac{1}{N} \sum_{n=1}^N \bbE_{\bx^{(n)}|\bY^{(n)},\Theta'}\|\bx^{(n)}-\bx^{(n)}\bC\|_2^2 + \lambda \| \bC \|_{l_1/F}  \\
        = & \frac{1}{N} \sum_{n=1}^N \left( \|\hbu_{\Theta', \bY^{(n)}} - \hbu_{\Theta', \bY^{(n)}} \bC\|^2_2 + {\rm tr}((\bI-\bC)^T \hat{\bSigma}_{\Theta'} (\bI-\bC)) \right) + \lambda \|\bC\|_{l_1/F}\\
        \doteq &\  G_n(\bC, \Theta') + \lambda \|\bC\|_{l_1/F} \\
        {\rm s.t.} \quad & h(\bW) = 0,
    \end{split}
    \end{equation}
where $\bC^{K}$ is a $\sum_j{K_j} \times \sum_j {K_j}$ matrix with its $(j,j')$ block as $\bC^{K}_{(j,j')} = \bC^{K}_{j'j}$, $\bC^{L}$ is a $L \times L$ matrix with its $(j,j')$ block as $\bC^{L}_{(j,j')} = \bC^{L}_{j'j}$, $\bC_{j'j} = \bC^L_{j'j} \otimes \bC^K_{j'j}$.
  
    We can convert Eq. (\ref{equ:M-step-c1}) into an unconstrained problem using the Lagrangian dual method:
  \begin{equation} \label{equ:Lagrangian}
      \hat{\bC}^K, \hat{\bC}^L \in \underset{\bC^K, \bC^L}{\arg \min} \underset{b>0}{\max}\  \tilde{G}_n(\bC, \Theta') + \lambda \| \bC \|_{l_1/F},
  \end{equation}
  where
  \begin{equation*}
      \tilde{G}_n(\bC, \Theta') = G_n(\bC, \Theta') + b h(\bW) + \frac{a}{2} h(\bW)^2.
  \end{equation*}
$b \in \mathbb{R}$ is dual variable and $a \in \mathbb{R}$  is the coefficient for quadratic penalty. We solve the Lagrangian dual problem by the dual ascent method. Due to the non-smoothness of $l_1/F$ norm, we use the proximal gradient method for group lasso penalty.
We summarize the algorithm in Algorithm \ref{alg:Largrangian}. 
\begin{algorithm}[t]
\caption{EM algorithm}
\label{alg:EM algorithm}
\begin{algorithmic}
   \STATE {\bfseries Input:} data $\calY$, tolerances $\epsilon_0$
   \STATE Initialize $\Theta^{(0)}=\mathrm{vec}( \bB^{(0)},\bC^{K (0)}, \bC^{L (0)}, \bR^{ (0)}, \omega^{2 (0)})$, $s = 0$.
   \REPEAT
    \STATE $s \leftarrow s + 1$
   \STATE $\hbu^{(n)}_{\Theta^{(s-1)},\bY^{(n)}} , \hat{\bSigma}_{\Theta^{(s-1)}}\leftarrow \text{Forward filtering \& backward smoothing} (\Theta^{(s-1)})$ via Appx. \ref{app:Expectation}
   \STATE $\bB^{(s)} \leftarrow \underset{\bB}{\arg \min}\ F_n(\bB,\Theta^{(s-1)})$ via Polar decomposition.
   \STATE $\bC^{K(s)},\bC^{L(s)} \leftarrow \underset{\bC^K,\bC^L}{\arg \min}\ G_n(\bC;\Theta^{(s-1)}) + \lambda \| \bC \|_{l_1/F}$ via Algorithm \ref{alg:Largrangian}
   \STATE Update $\omega^{2 (s)}_0$ by Eq. (\ref{equ:m-step-c3})
   \STATE $\Theta^{(s)} \leftarrow \mathrm{vec}(\bB^{(s)},\bC^{K(s)},\bC^{L(s)}, \bR^{(s)}, \omega^{2 (s)}_0)$
   \UNTIL{$D(\Theta^{(s)}, \Theta^{(s-1)})<\epsilon_{0}$}
\end{algorithmic}
\end{algorithm}

    After obtaining the transition matrix $\hat{\bC}$, $\hat{\omega}_0^2$ can be solved in a closed form as
    \begin{equation} \label{equ:m-step-c3}
    \begin{split}
        \hat{\omega}_0^2 =  \frac{1}{NM} \sum_{n=1}^N \bigg( \|\hbu_{\Theta', \bY^{(n)}} - \hbu_{\Theta', \bY^{(n)}} \hat{\bC}_j\|_2^2
        + {\rm tr}((\bI-\hat{\bC})^T \hat{\bSigma}_{\Theta'} (\bI-\hat{\bC}))\bigg).
    \end{split}
    \end{equation}
    Combine Eqs. (\ref{equ:M-step-u1}), (\ref{equ:M-step-u2}), (\ref{equ:M-step-c1}) and (\ref{equ:m-step-c3}), we can update $\Theta$ and replace $\Theta'$ by the updated $\Theta$. 
    
    We repeat the Expectation-step and Maximization-step iteratively until convergence, i.e., the difference between the estimated parameters
    \begin{equation*}
        D(\Theta, \Theta') = \sqrt{\| \bC - \bC'\|_F^2  +  \| \bB - \bB'\|_F^2 + \| \br - \br'\|_2^2 + \| \omega_0^2 - \omega_0^{2'}\|_F^2}
    \end{equation*}
    is smaller than a threshold $\epsilon_0$.
    We summarize the regularized EM algorithm in Algorithm \ref{alg:EM algorithm}, where $\Theta = \{\bC,\bB,\br,\omega_0^2\}$ and $\Theta' = \{\bC',\bB',\br',\omega_0^{2'}\}$.
    \begin{algorithm}[t]
\caption{Algorithm for Largrangian dual problem}
\label{alg:Largrangian}
\begin{algorithmic}
    \STATE {\bfseries Input:} posterior distribution $\hbu^{(n)}_{\Theta^{(s-1)},\bY^{(n)}} , \hat{\bSigma}_{\Theta^{(s-1)}}$, tolerance $h_{tol}$, learning rate $lr$, $\gamma$.
    \STATE Initialize $\bC^{K}, \bC^{L}, a \leftarrow 1, b \leftarrow 0$.
    \REPEAT
    \STATE Update $\bC^{K}, \bC^{L}$ by minimizing $\tilde{G}_n$ by gradient method. 
    \STATE $b \leftarrow b + ah(\bW)$
    \STATE $a \leftarrow lr * a$
    \UNTIL {$h(\bW) < h_{tol}$}
    \FOR{$i$ from $1$ to $P$}
    \FOR{$j$ from $1$ to $P$}
    \IF{$\|\bC_{ij}\|_F > \gamma \lambda$}
        \STATE $\bC^{L}_{ij} \leftarrow \bC^L_{ij} - \gamma \lambda \frac{\bC^{L}_{ij}}{\| \bC^{}_{ij} \|_F}$
    \ELSE
        \STATE $\bC^{L}_{ij} \leftarrow \mathbf{0}$
    \ENDIF
    \ENDFOR
    \ENDFOR
\end{algorithmic}
\end{algorithm}
    

    
    \section{Theoretical Properties}
     \label{sec:theory}
         In the following, we prove that when certain model assumptions hold, the estimated parameters can converge to those of the true model. Assuming for the true model, its parameter set is denoted as $\Theta^{*} = \{\bC^*, \bB^*, \omega_0^{2 *}, \br^*\}$. Condition \ref{ass:bound} gives the upper and lower bounds of the variances that ensure the data covariance matrix is not degenerate, i.e.,  
     \begin{condition} \label{ass:bound}
    All the eigenvalues of $\bSigma^* = (\bI - \bC^*)^{-T}\omega_0^{2 *}(\bI - \bC^*)^{-1}$ should be greater than a constant $\eta_{\bSigma^*} > 0$ and finite, and $\forall j\in 1,\ldots P$ and $l=1, \ldots, L_j$, we assume $r_{jl}^2<\infty$.
    \end{condition}        

 Condition \ref{ass:num-latent-bound} ensures the identifiability of decomposition on $\bY$. 
    \begin{condition} \label{ass:num-latent-bound}
        The number of latent variables is smaller than the number of sampling points for each function, i.e., $K_j < T, \forall j=1,\ldots, P$.
    \end{condition}

By combining these three conditions, we can obtain the good property for all the models in the equivalent class $\mathfrak{D}$ in Theorem \ref{the:-equivalence-class}. 

    \begin{theorem}[Equivalence class] \label{the:-equivalence-class}
        Define the equivalence class of the true parameters $\Theta^*$ as $\mathfrak{D}$. Under Conditions \ref{ass:equal-variance} to \ref{ass:num-latent-bound},  for any parameters $\Theta^e = \{\bC,\bB,\br,\omega_0^2 \} \in \mathfrak{D}$, it can be represented by the following form:
        \begin{align*}
            \bB_j &= \bB^*_j \bQ_j, \\
            \br_{jl} &= \br^*_{jl}, \\
            \omega_0^2 &= \omega_0^{2 *}, \\
            \bC_{j'jl'l} &= \bQ_j \bC^*_{j'jl'l} \bQ_{j'}^T,
        \end{align*}
 where $\bQ_j \in \mathbb{R}^{K_j \times K_j}$ is an orthogonal matrix satisfying $\bQ_j^T \bQ_j = \bQ_j \bQ_j^T = \bI$. This states that the equivalence class of the true solution is only the set of orthogonal transformations of $\Theta^*$. 
    \end{theorem}
    \begin{proof}
        The proof is in Appx. \ref{pro:-equivalence-class}
    \end{proof}
Intuitively speaking, this indicates though we choose different orthogonal basis functions to map $\bY_{jl}$, the spaces spanned by these orthogonal functional basis spaces are the same. Therefore, we can obtain the true causal structure once we get any equivalent solution $\Theta^e \in \mathfrak{D}$.
    
    Next, we aim to prove that our regularized EM algorithm is capable of discovering the true causal order and parameters when the initial parameters $\Theta^{(0)}$ are close to the true parameters $\Theta^*$. 
    We give the definition of population analogs of $F_n$ and $G_n$ in Definition \ref{def-QFG}. 
    \begin{definition}[Population analogs] \label{def-QFG}
    Define $F$ and $G$ as the population analogs of $F_n$ and $G_n$ respectively, i.e., 
   \begin{align*}
    F(\bB,\Theta') &= \int \sum_{j=1}^P \sum_{l=1}^{L_j} \mathbb{E}_{\bx|\bY;\Theta'} \|\bY_{jl}-\bB_j\bx_{jl}\|^2_2 p(\bY;\Theta^*) {\rm d}\bY, \\
        G(\bC, \Theta') &= \int \mathbb{E}_{\bx|\bY;\Theta'}\| \bx - \bx \bC\|_2^2 p(\bY;\Theta^*) {\rm d}\bY.
    \end{align*}
Using the Strong Law of Large Numbers, we can observe that as $n$ approaches infinity, the results $F_n$ and $G_n$ converge almost surely to $F$ and $G$ respectively.        
    \end{definition}
    Theorem \ref{the:identifiability} shows the true parameters $\Theta^*$ can maximize the population log-likelihood function and satisfy the self-consistency property \citep{mclachlan2007algorithm}.
    \begin{theorem}[Self-consistency] \label{the:identifiability}
        When Conditions \ref{ass:equal-variance} and \ref{ass:bound} hold, we can obtain $\Theta^*$ by minimizing $G(\cdot,\Theta^*)$ and $F(\cdot, \Theta^*)$. 
    \end{theorem}
    \begin{proof}
        The proof is in Appx. \ref{pro:identifiability}. 
    \end{proof}
Next, we introduce the theorem related to causal structure. We define the causal order $\pi$ in Definition \ref{def:causal-order}.
    \begin{definition} \label{def:causal-order}
Since $\bW$ is the adjacency matrix of a DAG $\mathcal{G}$, the nonzero entries of $\bW$ define the causal order of graph $\pi \in \mathbb{S}_P$, which can be represented by a permutation over $1,2,...,P$. $\pi(i)$ represents the position of node $i$ in the order. A causal order $\pi$ is consistent with a DAG $\mathcal{G}$ if and only if:
    \begin{align}
        \bW_{ij} \ne 0 \Rightarrow \pi(i) < \pi(j).
    \end{align}  
    \end{definition}
With abusive use of notation, we denote $\bC(\pi)$ to address this $\bC$ is consistent with causal order $\pi$. Then define $\calC(\pi)$ as the set of $\bC(\pi)$ that has the same causal order $\pi$, i.e., $\bC(\pi)\in \calC(\pi)$. Denote $\bC^*_{\Theta}(\pi) = \underset{\bC(\pi) \in \calC(\pi)}{\arg \min} G(\bC(\pi), \Theta)$. Let $\Pi_0^*$ be the set of all causal orders consistent with $\bC^*$. Since Theorem \ref{the:identifiability} holds, we have $\bC^*_{\Theta^*}(\pi_0) = \bC^*, \forall \pi_0 \in \Pi_0^*$. 
    \begin{condition}[Omega-min] \label{ass:omega-min}
        Under Conditions \ref{ass:equal-variance} and \ref{ass:bound}, for all $\pi \notin \Pi_0^*, \exists \eta_1 > 0$ that:
        \begin{equation}
            G(\bC^*, \Theta^*) - G(\bC^*_{\Theta^*}(\pi), \Theta^*) < - \eta_1.
        \end{equation}
    \end{condition}
  
    Condition \ref{ass:omega-min} assumes that if we restrict our model to a wrong causal order $\pi' \notin \Pi_0^*$, $G(\bC^*_{\Theta^*}(\pi'),\Theta^*)$ will increase by at least $\eta_1$. This is similar to the Omega-min condition in \citet{van2013ell_}, 
    and is used to justify the precision of our true model.
    

    \begin{lemma} \label{lem:extend-bound}
        Under Condition \ref{ass:bound} and \ref{ass:num-latent-bound}, $\exists \tr_1$, the following inequalities hold for $\Theta \in \mathbb{B}_2(\Theta^*,\tr_1)$:
        
        (1) $\underset{\Theta \in \mathbb{B}_2(\Theta^*,\tr_1)}{\max} \bbE_{\bY} \bbE_{\bx|\bY;\Theta} (\|\bx\|_2^8) < \infty$; 
        

        (2)  $\underset{\Theta \in \mathbb{B}_2(\Theta^*,\tr_1)}{\max} \underset{j,l}{\max} \ \bbE_{\bY} \bbE_{\bx|\bY;\Theta} (\| \bY_{jl} \hbu_{jl, \Theta, \bY}^T  \|_F^4) < \infty$;
        
        (3) $\underset{\Theta \in \mathbb{B}_2(\Theta^*,\tr_1)}{\min} \text{\rm minEig}({\rm Cov}(\hbu_{\Theta^*, \bY}) + \hat{\bSigma}_{\Theta^*}) > 0$;
            
        (4) $\underset{\Theta \in \mathbb{B}_2(\Theta^*,\tr_1)}{\min} \underset{j,l}{\min} \ \sigma_{\min} (\bbE_{\bY} \bbE_{\bx|\bY;\Theta}(\bY_{jl} \hbu_{jl, \Theta, \bY}^T ))> 0$;
        
where ${\rm minEig}(\cdot)$ is the minimum eigenvalue of the matrix. $\sigma_{\min}(\bA)$ is $k$-th maximum singular value of the matrix for $\bA \in \mathbb{R}^{T\times k}$, where $\sigma_{\min}(\cdot) > 0$ shows that the matrix is column full rank. $\mathbb{B}_2(\Theta^*,r):= \{\Theta|D(\Theta, \Theta^*) \leq r\}$.
    \end{lemma}
    \begin{proof}
        The proof is in Appx. \ref{pro:extend-bound}.
    \end{proof}
        Lemma \ref{lem:extend-bound} shows that the posterior distribution $p(\bx|\bY;\Theta)$ is not degraded and has bounded variance when $\Theta \in \mathbb{B}_2(\Theta^*,\tr_1)$.
        We denote
        \begin{align*}
            \underset{\Theta \in \mathbb{B}_2(\Theta^*,\tr_1)}{\sup} & \bbE_{\bY} \bbE_{\bx|\bY;\Theta} (\|\bx\|_2^8) = \mathtt{x}_{\sup}^8, \\
            \underset{\Theta \in \mathbb{B}_2(\Theta^*,\tr_1)}{\sup} & \underset{j,l}{\sup} \ \bbE_{\bY} \bbE_{\bx|\bY;\Theta} (\| \bY_{jl} \hbu_{jl, \Theta, \bY}^T  \|_F^4) = \mathtt{y}_{\sup}^4, \\
            \underset{\Theta \in \mathbb{B}_2(\Theta^*,\tr_1)}{\inf} & \text{\rm minEig}({\rm Cov}(\hbu_{\Theta, \bY}) + \hat{\bSigma}_{\Theta}) = \mathtt{s}_{\inf}, \\
            \underset{\Theta \in \mathbb{B}_2(\Theta,\tr_1)}{\inf} & \underset{j,l}{\min} \ \sigma_{\min} (\bbE_{\bY} \bbE_{\bx;\Theta}(\bY_{jl} \hbu_{jl, \Theta, \bY}^T )) = \mathtt{b}_{\inf},
        \end{align*}
 where $\mathtt{x}^8_{\sup}, \mathtt{y}^4_{\sup}, \mathtt{s}_{\sup}, \mathtt{b}_{\inf} > 0$ are universal constants depended on $\tr_1$. 

    \begin{lemma} \label{lem:extend-omega-min}
    Under Condition \ref{ass:omega-min}, $\exists \tr_2, \forall \Theta \in \mathbb{B}_2(\Theta^*, \tr_2)$, denote $\Pi_\Theta^* = \{  \pi|\pi = \arg\min_{\pi'} G(\bC^*_\Theta(\pi'), \Theta)$\} and $\bC^*_{\Theta} = \underset{\bC^*_\Theta(\pi)}{\arg \min} G(\bC^*_\Theta(\pi);\Theta)$. We have:
    
    (1) $\Pi_{\Theta}^* = \Pi_0^*$, 
    
    (2) For all $\pi \notin \Pi_{\Theta}^{*}, \exists 0<\eta_2<\eta_1$ that:
    \begin{equation}
         G(\bC^*_{\Theta}, \Theta) - G(\bC^*_{\Theta}(\pi), \Theta) < - \eta_2.
    \end{equation}
\end{lemma}
    \begin{proof}
        The proof is in Appx. \ref{pro:extend-omega-min}.
    \end{proof}
    
    Lemma \ref{lem:extend-omega-min} extends Condition \ref{ass:omega-min} from $\Theta^*$ to all $\Theta \in \mathbb{B}_2(\Theta^*, \tr_2)$. It states that when $\Theta$ is close to $\Theta^*$, we can still identify the true causal order by minimizing $G(\bC,\Theta)$. Taking $\tr=\min(\tr_1,\tr_2)$, Lemma \ref{lem:error-bound} and \ref{lem:polar-decomposition} provide the lower bound for the error when estimating $\hat{\bC}$ and $\hat{\bB}$, in a single iteration of the regularized EM iteration. 

    \begin{lemma} \label{lem:error-bound}
    Under Conditions \ref{ass:equal-variance}, \ref{ass:bound}, \ref{ass:omega-min} and suppose that we solve the optimization of Eq. (\ref{equ:M-step-c1}) with specified regularization parameters $\lambda$ and $\Theta \in \mathbb{B}_2(\Theta^*, \tr)$. Since $\mathbb{B}_2(\Theta^*,\tr)$ is a contact set, we denote $ \mathtt{c}_{\sup} \doteq \underset{\Theta \in \mathbb{B}_2(\Theta^*,\tr)}{\sup} \underset{\pi}{\sup} \| \bC^*_{\Theta}(\pi) \|_{l_1/F}$ and $\mathtt{d}_{\sup}^4 = \underset{\Theta \in \mathbb{B}_2(\Theta^*,\tr)}{\sup} \underset{\pi}{\sup} \| \bI - \bC^*_{\Theta}(\pi) \|_F^4$.
    If the following conditions are satisfied for $\varrho_1, \varrho_2, \varrho_3 \in (0,1), \delta_1 \in (0, 1/2)$:
    \begin{align*}
        & \eta_2 > 2 \sqrt{\frac{\mathtt{d}_{\sup}^4 \mathtt{x}_{\sup}^4}{\varrho_1 N}} + \lambda(2\delta_1 + 1)  \mathtt{c}_{\sup}, \\
        & \frac{\mathtt{d}_{\sup}^2 \mathtt{x}_{\sup}^4}{\lambda^2 N \delta_1^2} < 1, \\
        & 1 - 2\varrho_1 - P!M\varrho_2 - \varrho_3 > 0, \\
        & \mathtt{s}_{\inf} > \sqrt{\frac{\mathtt{x}_{\sup}^4}{N} + \sqrt{\frac{\mathtt{x}_{\sup}^8}{\varrho_3 N}}}.
    \end{align*}
  Denote $\hat{\bC}$ and $\hat{\pi}$ as the matrix and corresponding causal order by solving Eqs. (\ref{equ:M-step-c1}) with $\Theta'=\Theta$. Then the following statements hold true:

    (1) With probability at least $1 - 2\varrho_1 - P!M\varrho_2$, $\hat{\pi} \in \Pi_0^*$;

    (2) With probability at least $1 - 2\varrho_1 - P!M\varrho_2 - \varrho_3$,
    \begin{equation*}
        \|\hat{\bC} - \bC^*_{\Theta}\|_F^2 \leq \frac{2\sqrt{ \frac{\mathtt{d}_{\sup}^4 \mathtt{x}_{\sup}^4}{\varrho_1 N}} + \lambda(2\delta_1 + 1)  \mathtt{c}_{\sup}}{\mathtt{s}_{\inf} - \sqrt{\frac{\mathtt{x}_{\sup}^4}{N} + \sqrt{\frac{\mathtt{x}_{\sup}^8}{\varrho_3 N}}}}.
    \end{equation*}
\end{lemma}
    \begin{proof}
        The proof is in Appx. \ref{pro:error-bound}.
    \end{proof}
\begin{lemma} \label{lem:polar-decomposition}
     Under Condition \ref{ass:bound} and \ref{ass:num-latent-bound}, denote $\bB_{\Theta}^*$ as the matrix that minimizes $F(\cdot,\Theta)$ with  $\bB_{\Theta j}^{T*}\bB_{\Theta j}^{*} = \bI, \forall j$ and $\hat{\bB}$ as the optimal solution to $F_n(\cdot, \Theta)$ with $\hat{\bB}_j^T \hat{\bB}_j = \bI, \forall j$. Then if for $\varrho_4, \varrho_5 \in (0,1)$:
     \begin{align*}
         & 1 - P \varrho_4 - P\varrho_5 > 0 , \quad {\rm and} \quad \mathtt{b}_{\inf} - \sqrt{\frac{\mathtt{y}^2_{\sup}}{N} + \sqrt{\frac{ \mathtt{y}^4_{\sup}}{\varrho_5 N}}} > 0,
     \end{align*}
 with probability at least $1 - P\varrho_4 - P\varrho_5$, we have:
    \begin{equation*}
        \|\hat{\bB} - \bB^*_{\Theta}\|_F^2 \leq \frac{P\left(\frac{\mathtt{y}^2_{\sup}}{N} + \sqrt{\frac{ \mathtt{y}^4_{\sup}}{\varrho_4 N}} \right) }{\left(\mathtt{b}_{\inf} - \sqrt{\frac{\mathtt{y}^2_{\sup}}{N} + \sqrt{\frac{ \mathtt{y}^4_{\sup}}{\varrho_5 N}}}\right)^2}.
    \end{equation*}
\end{lemma}
    \begin{proof}
        The proof is in Appx. \ref{pro:polar-decomposition}.
    \end{proof}
Next, we aim to derive an upper bound for the total error bound of our regularized EM algorithm, i.e., $D(\Theta^{(S)}, \Theta^*)$ for total $S$ iterations in the regularized EM algorithm. Under certain conditions (seeing Conditions \ref{ass:concavity-smoothness} and \ref{ass:Lipschitz-Gradient} in Appendix. \ref{app:converge}), Theorem \ref{the:convergence} and Corollary \ref{rem:asymptotic} establish the convergence properties and error analysis of our regularized EM algorithm. These results hold when the initial solution $\Theta^{(0)}$ is in proximity to the true solution, encompassing scenarios of both finite $N$ and as $N$ approaches infinity. Furthermore, this property still holds when replacing $\Theta^*$ with any equivalent solution $\Theta^{e} \in \mathfrak{D}$. As Theorem \ref{the:-equivalence-class} states, any $\Theta^e \in \mathfrak{D}$ has the same causal structure as $\Theta^*$. Therefore, we show that our regularized EM algorithm can effectively learn the correct causal structure locally.

\begin{theorem} \label{the:convergence}
   Assume Conditions \ref{ass:equal-variance} to \ref{ass:Lipschitz-Gradient} and the conditions in Lemmas \ref{lem:error-bound} and \ref{lem:polar-decomposition} are satisfied, and the EM estimator $M(\Theta) \doteq \underset{\Theta'}{\arg \max}\ Q(\Theta';\Theta) - \lambda \calR(\bC) $ is contractive with parameters $\kappa \in (0,1)$ in the ball $\mathbb{B}_2(\Theta^*, \tr)$. Denote $S$ as the total iterations of the regularized EM algorithm, we have 
        \begin{equation*}
         D(\Theta^{(S)}, \Theta^*) \leq \kappa^S D( \Theta^{(0)}, \Theta^*) + \frac{1}{1-\kappa} \epsilon(\delta / S, N / S, \tr),
     \end{equation*}
where 
    \begin{equation*}
        \delta / S = 2\varrho_1 + MP!\varrho_2 + \varrho_3 + P\varrho_4 + P\varrho_5 + M \varrho_6,
    \end{equation*}
     \begin{equation*}
         \epsilon(\delta / S, N / S, \tr) = \left( \frac{2\sqrt{ \frac{\mathtt{d}_{\sup}^4 \mathtt{x}_{\sup}^4 S}{\varrho_1 N}} + \lambda(2\delta_1 + 1)  \mathtt{c}_{\sup}}{\mathtt{s}_{\inf} - \sqrt{\frac{\mathtt{x}_{\sup}^4 S}{N} + \sqrt{\frac{\mathtt{x}_{\sup}^8 S}{\varrho_3 N}}}} + \frac{P\left(\frac{\mathtt{y}^2_{\sup}S}{N} + \sqrt{\frac{ \mathtt{y}^4_{\sup} S}{\varrho_4 N}} \right) }{\left(\mathtt{b}_{\inf} - \sqrt{\frac{\mathtt{y}^2_{\sup} S}{N} + \sqrt{\frac{ \mathtt{y}^4_{\sup} S}{\varrho_5 N}}}\right)^2}\right)^{1/2} + O((N/S)^{-1/2}).
     \end{equation*}
\end{theorem}
    \begin{proof}
According to Lemma \ref{lem:error-bound} and \ref{lem:polar-decomposition} above, together with Lemma \ref{lem:bound-r} and \ref{lem:bound-w}  which gives the error bound of $\br$ and $\omega^2_0$ in each EM iteration as $O((N/S)^{-1/2})$ with probability $1-M\varrho_6$, we can prove Theorem \ref{the:convergence} following the procedures in Theorem 5 in \citet{balakrishnan2017statistical}. 
    \end{proof}
    \begin{corollary}[Asymptotic property] \label{rem:asymptotic}
        Based on Theorem \ref{the:convergence}, we have the following two corollaries:
        
        (1) As $N \rightarrow \infty$, by setting $\lambda \sim N^{-1/2 + \nu}$ with $\nu \in (0, 1/2)$, the conditions in Lemma \ref{lem:error-bound} hold. Then $\epsilon(\delta/S, N/S, \tr) = \mathcal{O}((N / S)^{(-1 + 2 \nu) / 4})$, and the total estimation error after $S$ EM iterations can be described as $D(\Theta^{(S)}, \Theta^*) \leq \kappa^S D(\Theta^{(0)}, \Theta^*) + \mathcal{O}((N/S))^{(-1 + 2 \nu) / 4})$.

        (2) Under $S \rightarrow \infty$, $N \rightarrow \infty$ and $N/S \rightarrow \infty$, we have $\Theta^{(S)} \rightarrow \Theta^*$ with probability 1.
    \end{corollary}

\section{Numerical study}
\label{sec:numerical}
To evaluate the performance of our methodology and selection of $\lambda$, we apply our MultiFun-DAG to solve a synthetic graphical model.  We show the performance of our algorithm on tasks of different combinations $(N,P,L_0,K_0)$, where $\forall j=1,\ldots,P$, we have $L_j=L_0$ and $K_j=K_0$.

In each experiment, the graphs are generated by Erdös-Rényi random graph model, where the functional data of the different nodes have the same Fourier basis $\nu_1(t)$, $\nu_2(t)$, ..., $\nu_K(t)$:
\begin{equation*}
    \nu_k(t) = 
    \begin{cases}
        1, & k = 1, \\
        \cos(2 \pi u t), & k = 2u, \\
        \sin(2 \pi u t), & k = 2u + 1,
    \end{cases} \forall u \in \mathbb{Z}, u \geq 1.
\end{equation*}

By combining Eq. (\ref{equ:Y-likelihood}), Eq. (\ref{equ:FPCA}) and Eq. (\ref{equ:linear-SEM}), we can write the representation of each functional data. The generated transition matrix is $\bC_{j'j} = c_{j'j} \mathbf{1}_{L_{j'}\times L_j} \otimes \mathbf{I}_K$, where $c_{j'j}$ is independently and identically generated from a uniform distribution $\mathcal{U}(-2, 0.5)\cup(0.5,2)$. The variance of noise is set by $\omega_0^2=1$ and $r_{jl}^2=0.01, \forall j=1,\ldots,P$.

For model comparison, we select two methods from the literature and another two variants of our MultiFun-DAG.
The baselines compared in this paper are introduced below. Since they cannot be directly used for DAG with nodes as multivariate functions, we modify these methods by concatenating multivariate functions as long univariate functions for analysis. 
    \begin{itemize}
        \item \textbf{FDGM\_S}: The functional directed graph model proposed by \citet{sun2017functional}. To deal with multivariate functional data for each node, we concatenate $L_j$ functional data of each node as long functional data with $L_j*T$ observation points. 
        \item \textbf{FDGM\_G}: The functional directed graph model proposed by \citet{gomez2020functional}. To deal with multivariate functional data for each node, we concatenate $L_j$ functional data as long univariate functional data with $L_j*T$ observation points.
        \item \textbf{MFGM}: This baseline provides a two-stage method to model the multivariate functional DAG. It first implements FPCA for each node separately to obtain their PC scores. Then it treats these scores as $\bX$ and uses the same structural learning method as MultiFun-DAG to estimate the causal structure, i.e.,
        \begin{align*}
            \min_{\bC^K, \bC^L} &\frac{1}{N} \|\bX-\mathbf{XC}\|_F^2 + \lambda \|\mathbf{\bC}\|_{l_1/F}, \\
            {\rm s.t.} \quad & {\rm tr}(\exp(\bW \circ \bW)) - P = 0,
        \end{align*}
        where $\bW$ and $\lambda$ have the same meaning as our method.
        \item \textbf{NoTears}: It first implements FPCA for the functional data of each node, where all the nodes share a common set of $K$ bases. After FPCA, the causal relationships between each PC score of each original node are learned by NoTears \citep{zheng2018dags}. Then the causal relationships between all the PC scores from two nodes are merged as the final causal relationship between these two nodes. 
    \end{itemize}
    
In this experiment, we aim to test the effectiveness of different methods to recover the true DAG structures. For brevity, we give the F1 score of the arcs to represent model performance, i.e.,
\begin{equation*}
    \text{F1 score} = 2 * \frac{\text{Precision} * \text{Recall}}{\text{Precision} + \text{Recall}}.
\end{equation*}

In Fig. \ref{fig:MultiFun-f1}, we see that our MultiFun-DAG has the best performance among all the baselines. The performance increases as the number of samples $N$ increases. MFGM has a similar performance to MultiFun-DAG but performs worse when the number of function data $L_0$ increases. This justifies the importance of our joint estimation of $\bX$ and $\bC$. 
\begin{figure}[t]
    \centerline{\includegraphics[width=0.75\columnwidth]{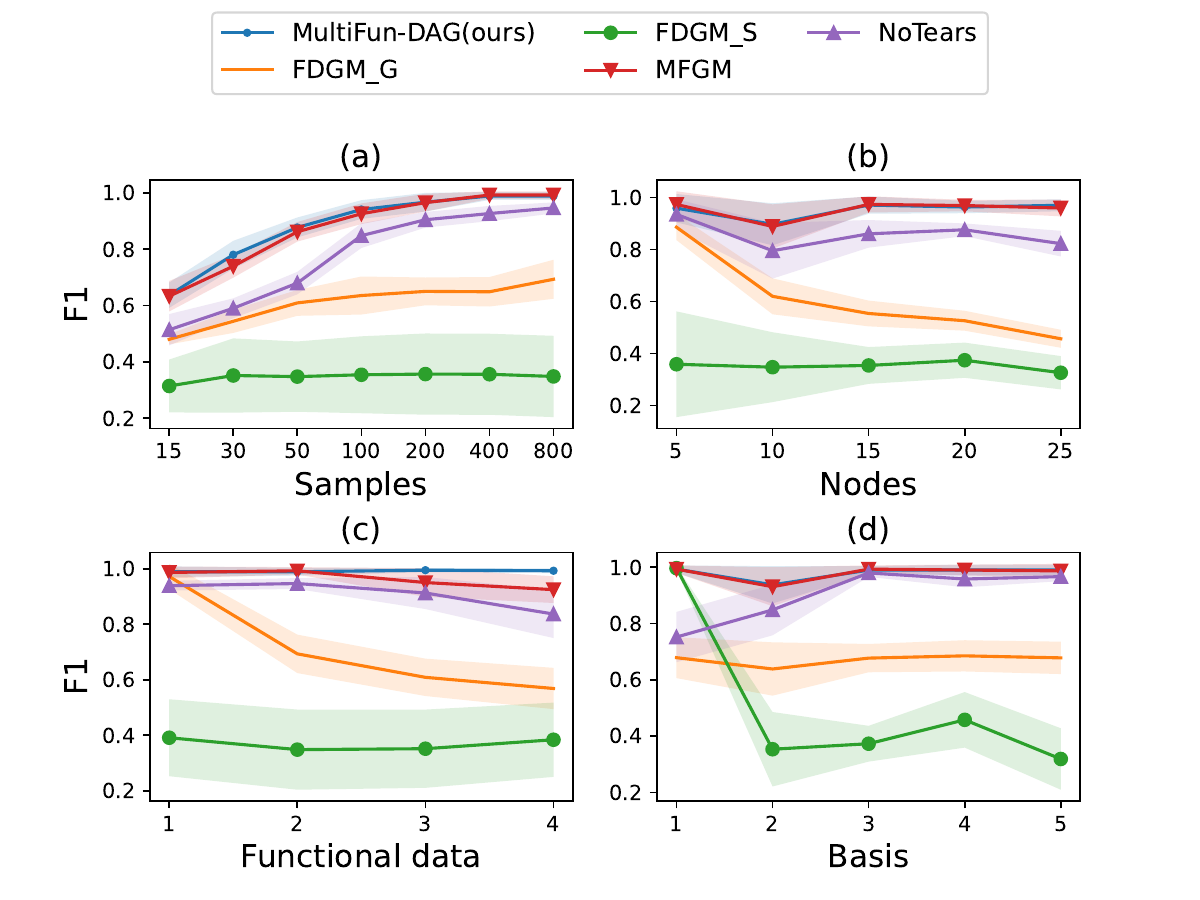}}
    \caption{F1 score of the edges with 95$\%$ confidence intervals: F1 score with different numbers of (a) samples $N$; (b) nodes $P$; (c) functions $L_0$; (d) bases $K_0$.} 
    \label{fig:MultiFun-f1}
    \end{figure}
 
Meanwhile, by comparing MFGM with NoTears, we verify the benefit of learning the DAG with vector-value nodes over the DAGs with scalar-value nodes. The difference in performance between MFGM and NoTears increases as the number of nodes increases or the number of functions increases. 
This is due to model complexity, i.e., the search space of causal order in NoTears is much larger than that in MFGM.
\begin{figure}[h]
    \centerline{\includegraphics[width=0.8\columnwidth]{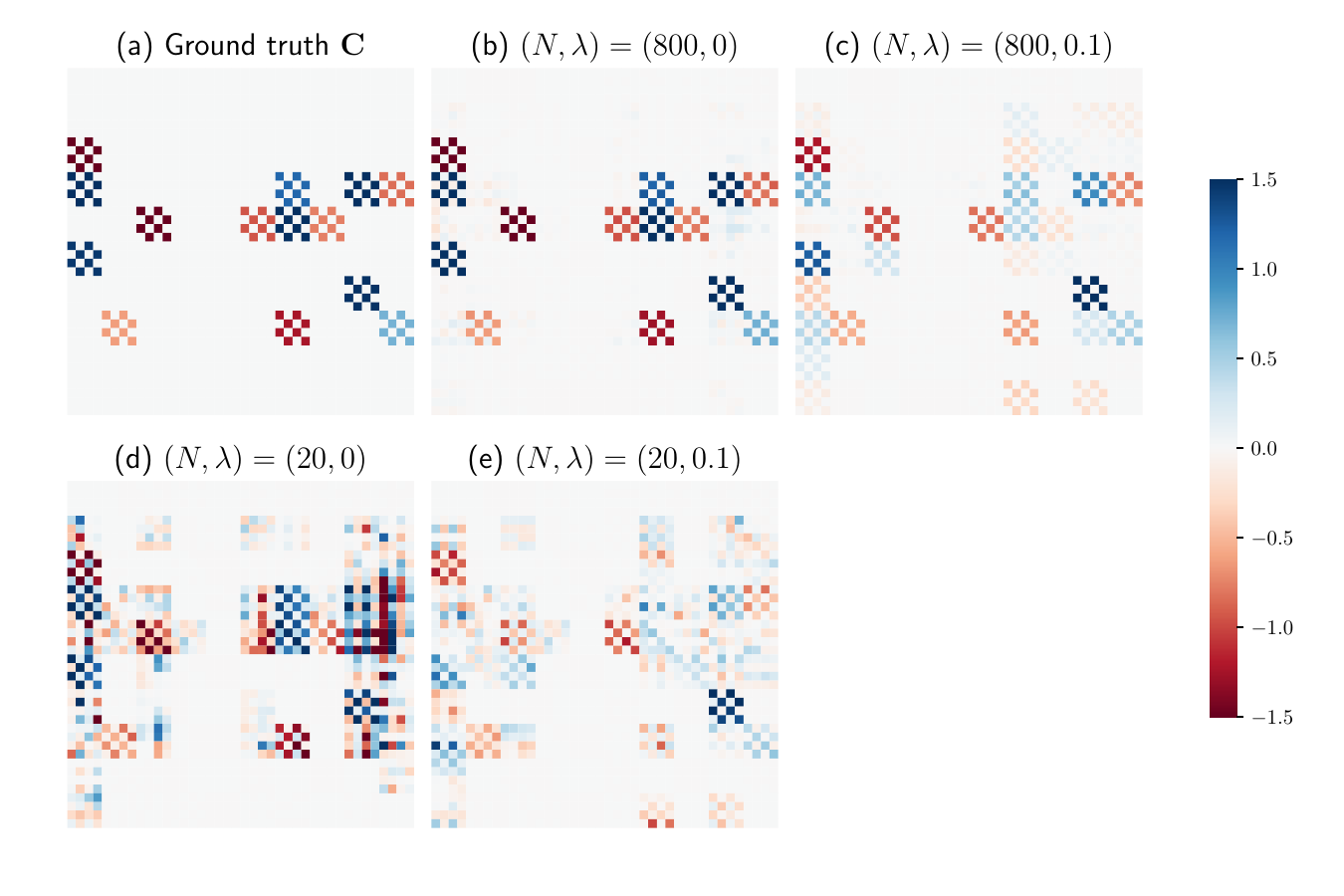}}
    \caption{Heatmap of $\bC^*$ and the estimated $\tilde{\bC}$ by MultiFun-DAG. Titles of the subplots represent the results under different experiment settings of $(N,\lambda)$. }
    \label{fig:MultiFun-heatmap}
\end{figure}

Table \ref{tab:mse} compares the result of our estimated parameters and the true parameters. 
In this case, we rotate the matrix $\hat{\bB}$ to the true matrix $\bB^*$. The rotation equation is given by Theorem \ref{the:-equivalence-class}: $\bB^*_j = \hat{\bB}_j \bQ_j$ and $\tilde{\bC}_{jk} = \bC_{jk}^L \otimes (\bQ_j^T \bC_{jk}^K \bQ_k) $, where $\bQ_j$ is an orthogonal matrix. The rotation process maintains the structure of DAG. Then we compare $\tilde{\bC}$ and $\bC^*$, by $\|\tilde{\bC}-\bC^*\|_F^2$. Furthermore, ${\rm MSE}_{\text{est}}$ and ${\rm MSE}_{\text{true}}$ measures the $l_2$ loss of $\calY$ for the estimated model and the true model, which can be computed by:
    \begin{align*}
        {\rm MSE}_{\text{est}} &= \frac{1}{NLT}\sum_{n=1}^{N} \sum_{j=1}^P \sum_{l=1}^{L_j} \bbE_{\bx^{(n)}|\bY,\hat{\Theta}}\|\bY^{(n)}_{jl} - \hat{\bB}\hat{\bx}^{(n)}_{jl}\|_2^2, \\
        {\rm MSE}_{\text{true}} &= \frac{1}{NLT}\sum_{n=1}^{N} \sum_{j=1}^P \sum_{l=1}^{L_j} \|\bY^{(n)}_{jl} - \bB^*\bx^{(n)}_{jl}\|_2^2.
    \end{align*}
    When ${\rm MSE}_{\text{est}}<{\rm MSE}_{\text{true}}$, overfitting occurs. When ${\rm MSE}_{\text{est}}>{\rm MSE}_{\text{true}}$, underfitting occurs. A smaller $|{\rm MSE}_{\text{est}}-{\rm MSE}_{\text{true}}|$, which is denoted by $|\Delta|$, indicates a smaller difference between the estimated and true parameters. Besides, a smaller $N$ needs a larger $\lambda$ to prevent overfitting, and on the contrary, a larger $N$ needs a smaller $\lambda$ to prevent underfitting. This might be because large $\lambda$ increases the bias and robustness of our algorithm. Fig. \ref{fig:MultiFun-heatmap} visualizes the estimated $\tilde{\bC}$ (the structure) under different experiment scenarios and its ground truth. With $(N,\lambda)=(800, 0)$, we could faithfully recover the structure.

    \begin{table}[t]
        \centering
        \caption{Estimated parameters v.s. True parameters. 
        }
        \scalebox{0.8}{
        \begin{tabular}{cccccc}
        \toprule
             $N$ & $\lambda$ & $\|\tilde{\bC}-\bC^*\|_F^2$ & ${\rm MSE}_{\text{est}}$ & ${\rm MSE}_{\text{true}}$ & $|\Delta|$\\
        \midrule
             800 & 0 & 1.21 & 1.99 & 2.013 & 0.02\\
             800 & 0.1 & 47.20 & 2.29 & 2.013 & 0.28 \\
             20 & 0 & 238.40 & 0.99 & 2.006 & 1.02\\
             20 & 0.1 & 107.74 & 1.46 & 2.006 & 0.55\\
        \bottomrule
        \end{tabular}
        }
        \label{tab:mse}
    \end{table} 
\section{Case study}
\label{sec:case}
In this section, we illustrate how our method can be applied to real-world urban traffic data for root cause analysis of traffic congestion. We focus on three types of traffic variables (nodes). (1) The real-time \textbf{traffic setting variables}, such as the real-time Origin-Destination (OD) demand, turning probability, the cycle time of the traffic light, etc., denoted as $\bS(t) = [\bS_1(t), \bS_2(t), ..., \bS_{P_s}(t)]$. 
(2) The real-time \textbf{traffic condition variables}, such as the occupancy of each lane, the average speed of each lane, the average waiting time of each lane, the number of vehicles in each lane, the number of halting vehicles in each lane, etc., 
denoted as $\mathbf{Y}(t) = [\mathbf{Y}_1(t), \mathbf{Y}_2(t), ..., \mathbf{Y}_{P_y}(t)]$. 
(3) The real-time traffic \textbf{congestion root cause variables}, such as long/short cycle time of traffic lights, phase imbalance, irrational guide lane, irrational phase sequence, imbalance of entrance, etc., 
denoted as $\bR(t)=[\bR_{1}(t),\ldots,\bR_{P_r}(t)]$.
Table \ref{tab:Abbreviation-func} summarizes the abbreviations and descriptions of each node.

We use the Simulation of Urban MObility (SUMO) \citep{krajzewicz2002sumo} to synthesize the real-time traffic data. We collect data from $\bS(t)$ and $\mathbf{Y}(t)$ every five minutes and simulate for 60 minutes. Therefore, each functional data has $T=12$ observation points. For each node of $\bY_{j}(t), t= 1,\ldots,T$, it has four functions, defined as $\bY_{j}\in \mathbb{R}^{4\times T}, j = 1, \ldots, P_y$. For each node of $\bS_{j}(t)$ and $\bR_{j}(t)$, it is a univariate function, defined as $\bS_{j}\in \mathbb{R}^{T}, j = 1, \ldots, P_s$ and $\bR_{j}\in \mathbb{R}^{T},  j = 1, \ldots, P_r$.  We set $\bR_j(t)\in\{0,1\}$. Here $\bR_j(t)=1$ indicates that the $j$-th type of congestion appears at time $t$, which is decided by rule-based algorithms in transportation. Its data is also collected every five minutes, with the same sampling grids as the other two types of traffic variables.
\begin{table}[h]
\caption{Abbreviation and the description of traffic data}
\label{tab:Abbreviation-func} 
\vskip 0.15in
\begin{center}
\scalebox{0.8}{
\begin{tabular}{ccc}
\toprule
Node & Name & Description \\
\midrule
$\bS_1\in \mathbb{R}^{T}$ & OD-A & OD demand of all direction \\
$\bS_2\in \mathbb{R}^{T}$ & OD-S & OD demand of certain direction \\
$\bS_3\in \mathbb{R}^{T}$ & T-A & Turning probability of all direction \\
$\bS_4\in \mathbb{R}^{T}$ & T-S & Turning probability of certain direction \\
$\bS_5\in \mathbb{R}^{T}$ & CT & Cycle time of traffic light \\ \midrule
$\mathbf{Y}_1 \in \mathbb{R}^{4 \times T}$ & OC & Occupancy of each of 4 lanes\\
$\mathbf{Y}_2 \in \mathbb{R}^{4 \times T}$ & MS & Mean speed of each of 4 lanes \\
$\mathbf{Y}_3 \in \mathbb{R}^{4 \times T}$ & MW & Mean waiting time of each of 4 lanes \\
$\mathbf{Y}_4 \in \mathbb{R}^{4 \times T}$ & NV & \# of vehicles in each of 4 lanes \\
$\mathbf{Y}_5 \in \mathbb{R}^{4 \times T}$ & NH & \# of halting vehicles in each of 4 lanes \\
\midrule
$\bR_1\in \mathbb{R}^{T}$ & Cycle-L & Long cycle time of traffic light \\
$\bR_2\in \mathbb{R}^{T}$ & Cycle-S & Short cycle time of traffic light \\
$\bR_3\in \mathbb{R}^{T}$ & Phase-imb & Phase imbalance\\
$\bR_4\in \mathbb{R}^{T}$ & lanes-irr & Irrational guide lane \\
$\bR_5\in \mathbb{R}^{T}$ & Entrance-imb & Imbalance of entrance \\
$\bR_6\in \mathbb{R}^{T}$ & Cycle-irr & Irrational phase sequence  \\
\bottomrule
\end{tabular}}
\end{center}
\vskip -0.1in
\end{table}

In the experiment, we set 11 levels on $\bS_1$, 3 levels on $\bS_2$, 3 levels on $\bS_3$, 4 levels on $\bS_4$ and 4 levels on $\bS_5$. Therefore, we have $11 \times 3 \times 3 \times 4 \times 4=1584$ treatment combinations. We run a single experiment on each treatment. In each experiment of $\bS$, we collect the traffic situation variables $\bY$ and the congestion indicator variables $\bR$, and treat them as one sample $[\bS^{(n)},  \mathbf{Y}^{(n)},\bR^{(n)}]$ for $n=1,2,...,1584$. 

Then we use MultiFun-DAG to learn the causal relationships between traffic setting variables and traffic congestion root cause variables. Based on domain knowledge, traffic setting variables have effects on the root cause variables, and different types of root cause variables will affect traffic condition variables. Therefore we assume the one-way connection from $\bS$ to $\bR$ and from $\bR$ to $\mathbf{Y}$. Moreover, we assume that there are no interior edges between nodes in $\bS$ and nodes in $\mathbf{Y}$. However, we assume that some types of congestion will lead to other types of congestion, i.e., there can be interior edges between nodes in $\bR$. 

The causal relationships between the variables in MultiFun-DAG are illustrated in Fig. \ref{fig:case-study}, and the probability interpretations are provided. The explainable insights about traffic congestion can be derived. For example, the edges Lanes-irr $\rightarrow$ Phase-imb and Cycle-S indicate that the irrationality of the guide lane could lead to the imbalanced traffic flow in different traffic signal phases, with some directions having long traffic queues and relatively short phase cycle. 
Thus, the guide lane should be better planned and the cycle time should be extended. 
In reality, the conditional probability $P(\bR_i|\bS,\bY)$ could also be used to predict the root cause probability in reality. 


\begin{figure}[t]
    \vskip 0.2in
    \begin{center}
    \centerline{\includegraphics[width=0.8\columnwidth]{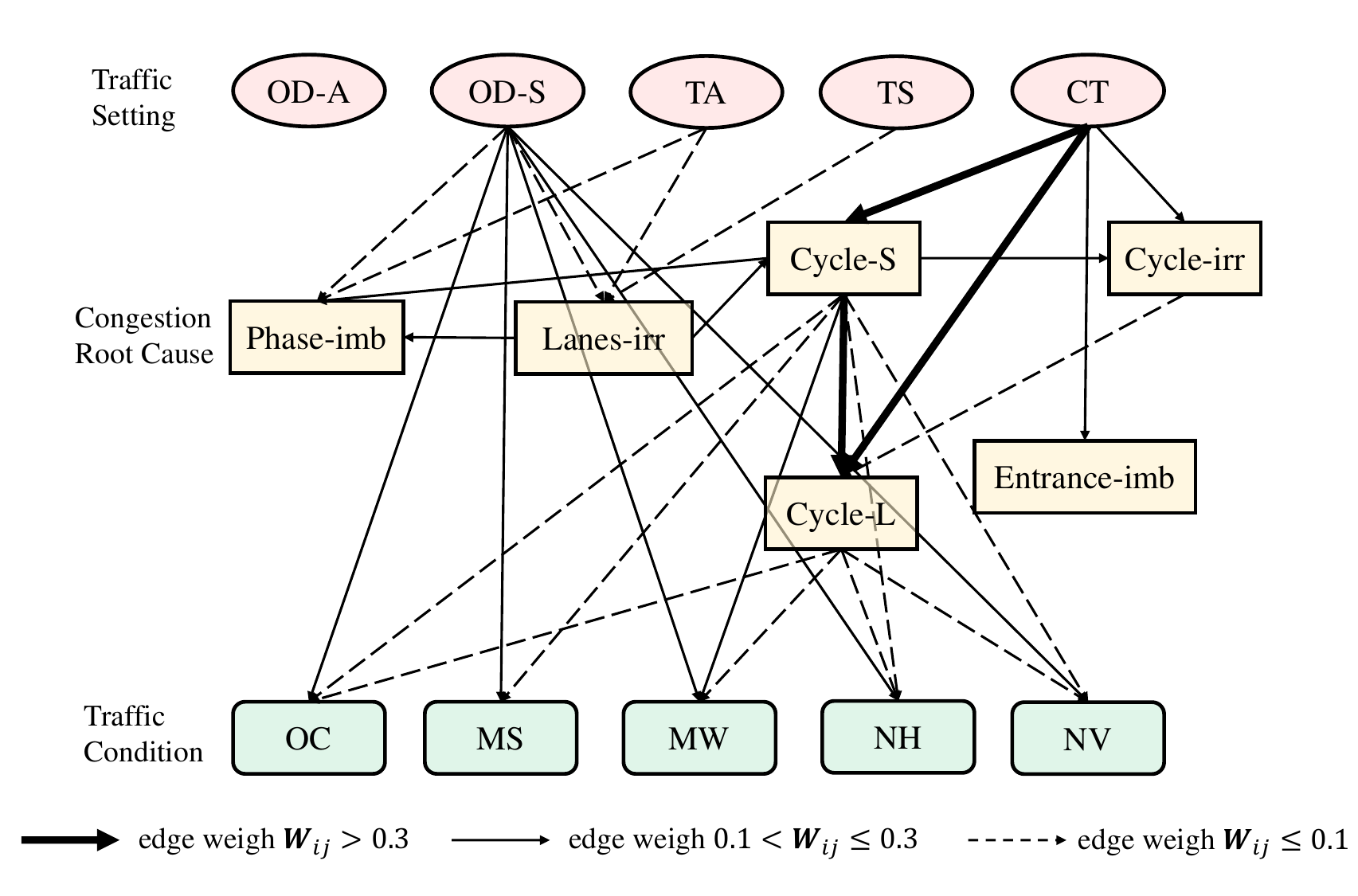}}
    \caption{The causal structure of traffic data.} 
    \label{fig:case-study}
    \end{center}
    \vskip -0.3in
\end{figure}

\section{Conclusion}
\label{sec:conclusion}
This paper presents a new framework for DAG with nodes as heterogeneous multivariate functional data. It simultaneously conducts functional decomposition for each node and uses the decomposition coefficients to represent the linear causal relationships between different nodes. By conducting a tailored regularized EM algorithm, the DAG structure together with other model parameters can be estimated based on a score-based structural learning algorithm with continuous acyclic constraint. The effectiveness of our algorithm is demonstrated by both theoretical proofs and numerical studies.
Some future works include extending the current MultiFun-DAG model to graphs with multi-mode data with both functional nodes and vector nodes. It is also interesting to conduct root causal analysis based on MultiFun-DAG for anomaly detection in multivariate functional data. 

\bibliographystyle{apalike}

\bibliography{main}
\newpage
\appendix
\section*{Appendix}
\section{Proof of theoretical property} 
\label{app:Identifiability}
\subsection{Proof of Theorem \ref{the:-equivalence-class}} \label{pro:-equivalence-class}
    \begin{proof}
        Denote $\Theta_1:=\{\bC^{(1)}, \bB^{(1)}, \br^{(1)}, \omega_0^{2(1)}\}$ and $\Theta_2:=\{\bC^{(2)}, \bB^{(2)}, \br^{(2)}, \omega_0^{2(2)}\}$ are two solution in the  equivalence class $\mathfrak{D}$. Denote $\bSigma ^{(1)} = (\bI-\bC^{(1)})^{-T}\omega_0^{2(1)}(\bI-\bC^{(1)})^{-1}$ and $\bSigma ^{(2)} = (\bI-\bC^{(2)})^{-T}\omega_0^{2(2)}(\bI-\bC^{(2)})^{-1}$ are the covariance matrices of $\bx$ determined by $\Theta_1$ and $\Theta_2$. Then the following equations hold true:
        \begin{align}
            \bB_j^{(1)} \bSigma_{jl,jl}^{(1)}          \bB_j^{(1) T} + r^{2 (1)}_{jl} \bI_T &= \bB_j^{(2)} \bSigma_{jl,jl}^{(2)} \bB_j^{(2) T} + r^{2 (2)}_{jl} \bI_T &\forall j,l \label{equ:-equivalence-class,samejl},  \\
            \bB_j^{(1)} \bSigma_{jl,jl}^{(1)}          \bB_{j'}^{(1) T}  &= \bB_j^{(2)} \bSigma_{jl,j'l'}^{(2)} \bB_{j'}^{(2) T} &\forall (j,l) \ne (j',l')\label{equ:-equivalence-class,difjl}.
        \end{align}
        
        For the Eq. (\ref{equ:-equivalence-class,samejl}), we have:
        \begin{align}
            \bB_j^{(1)} \bSigma_{jl,jl}^{(1)}          \bB_j^{(1) T} - \bB_j^{(2)} \bSigma_{jl,jl}^{(2)} \bB_j^{(2) T} &=  (r^{2 (2)}_{jl} - r^{2 (1)}_{jl}) \bI_T, \label{equ:-equivalence-class,samejl,p1} 
        \end{align}
        If $r^{2 (2)}_{jl} - r^{2 (1)}_{jl} \ne 0$ in Eq. (\ref{equ:-equivalence-class,samejl,p1}), the rank of the right-hand side is $T$, while the rank of the left-hand side is less than or equal to $K_j < T$, so the equation does not hold. Therefore, we have $r^{2 (2)}_{jl} - r^{2 (1)}_{jl} =0$, and $\bB_j^{(1)} \bSigma_{jl,jl}^{(1)} \bB_{j'}^{(1) T} = \bB_j^{(2)} \bSigma_{jl,j'l'}^{(2)} \bB_{j'}^{(2) T}, \forall j,j',l,l'$. This implies that $\bB_j^{(1)} = \bB_j^{(2)} \bQ_j$ with orthogonal matrix $\bQ_j$. From Eq. (\ref{equ:-equivalence-class,difjl}), we obtain $\bSigma_{jl,j'l'}^{(1)} = \bQ_j \bSigma_{jl,j'l'}^{(2)} \bQ^T_{j'}$.

        
        
        
        The optimality and uniqueness of the solution are proved in Lemma 5.1 in \cite{aragam2015learning} under the assumption of equal variances (Condition \ref{ass:equal-variance}). It is shown that for any given $\bSigma^{(1)}$, there exists a unique solution of $\bC^{(1)}$. We can show that for any $\bSigma^{(2)}$ satisfying $\bSigma_{jl,j'l'}^{(1)} = \bQ_j \bSigma_{jl,j'l'}^{(2)} \bQ^T_{j'}$, $\bC^{(2)}$ satisfying $\bQ_j \bC^{(2)}_{j'jl'l} \bQ_{j'}^T = \bC^{(1)}_{j'jl'l}$ is also the unique solution for $\bSigma^{(2)}$. 
    
    \end{proof}
\subsection{Proof of Theorem \ref{the:identifiability}}
\label{pro:identifiability}
\begin{proof}
    It is equivalent to prove that the optimal points to $F(\cdot, \Theta^*)$ and $G(\cdot, \Theta^*)$ are unique since $\hat{\br}$ and $\hat{\omega}_0^2$ are determined on $\hat{\bB}$ and $\hat{\bC}$. The uniqueness of $F(\cdot, \Theta^*)$ is guaranteed by the uniqueness of polar decomposition. As for $G(\cdot, \Theta^*)$, the uniqueness is proved by Lemma 5.1 in \citet{aragam2015learning}.
\end{proof} 

\subsection{Proof of Lemma \ref{lem:extend-bound}} \label{pro:extend-bound}
\begin{proof}
        Proposition. \ref{pro:represent} shows that the mean of posterior distribution $\hbu_{\Theta,\bY}$ can be represented by $\hbu_{\Theta,\bY} = \bA_{\Theta} {\rm vec}(\bY)$ and the covariance is represented by $\hat{\bSigma}_{\Theta}$. It is easy to show that $\bA_{\Theta}$ and $\hat{\bSigma}_{\Theta}$ are continuous functions of $\Theta$ by following the forward \& backward update in Appx. \ref{app:Expectation}. Therefore, (1) and (2) are hold. 

        For (3), from Lemma \ref{lem:post-true}, we have:
        \begin{equation*}
            {\rm minEig}({\rm Cov}(\hbu_{\Theta^*, \bY}) + \hat{\bSigma}_{\Theta^*}) = {\rm minEig}(\bSigma^*) > \eta_{\bSigma^*}.  
        \end{equation*}
        
        Because  $\bA_{\Theta}$ and $\hat{\bSigma}_{\Theta}$ are continuous for $\Theta$, for some $0 < \mathtt{s}_{\inf} < \eta_{\bSigma^*}$ and $\epsilon_1=\frac{1}{c_1}(\eta_{\bSigma^*}-\mathtt{s}_{\inf}), \exists \tr_a$ that $\forall \Theta \in \mathbb{B}_2(\Theta^*, \tr_a)$, we have:
        \begin{equation*}
            \|({\rm Cov}(\hbu_{\Theta^*, \bY}) + \hat{\bSigma}_{\Theta^*}) - ({\rm Cov}(\hbu_{\Theta, \bY}) + \hat{\bSigma}_{\Theta}) \|_F \leq c_1 \epsilon_1.
        \end{equation*}
        
        From Lemma \ref{lem:pertubation-eig}, we have
        \begin{equation*}
            |\text{\rm minEig}({\rm Cov}(\hbu_{\Theta, \bY}) + \hat{\bSigma}_{\Theta}) - \text{\rm minEig}({\rm Cov}(\hbu_{\Theta^*, \bY}) + \hat{\bSigma}_{\Theta^*}) | < c_1 \epsilon_1,
        \end{equation*}
        and we have:
        \begin{equation*}
            \text{\rm minEig}({\rm Cov}(\hbu_{\Theta, \bY}) + \hat{\bSigma}_{\Theta}) > \eta_{\bSigma^*}-c_1\epsilon_1> \mathtt{s}_{\inf}.
        \end{equation*}
        Then (3) is hold.
        
        For (4), $\forall j\in 1,\ldots,P $ and $l \in 1,\ldots,L_j$, we have
        \begin{equation*}
            \bbE_{\bY} \bbE_{\bx|\bY;\Theta^*}(\bY_{jl}\hbu_{jl, \bY, \Theta^*}^T) = \bbE_{\bx|\Theta^*}(\bB^*_j \bx_{jl} \bx_{jl}^T) = \bB^*_j\bSigma^*_{jl}
        \end{equation*}
        where $\bB^*_j\bSigma^*_{jl}$ is column full rank since $\bB^*_j$ is column full rank and $\bSigma^*_{jl}$ is full rank. Therefore, we have $\sigma_{\min}(\bB^*_j\bSigma^*_{jl})>0$. Because  $\bA_{\Theta}$ and $\hat{\bSigma}_{\Theta}$ are continuous to $\Theta$, for some $0 < \mathtt{b}_{\inf} < \min_{j,l}\sigma_{\min}(\bB^*_j\bSigma^*_{jl})$ and $\epsilon_2=\frac{1}{c_2}(\sigma_{\min}(\bB^*_j\bSigma^*_{jl})-\mathtt{b}_{\inf}), \exists \tr_{b,jl}$ that $\forall \Theta \in \mathbb{B}_2(\Theta^*, \tr_{b,jl})$, we have:
        \begin{equation*}
            \| \bbE_{\bY} \bbE_{\bx|\bY;\Theta^*}(\bY_{jl}\hbu_{jl, \bY, \Theta^*}^T) - \bbE_{\bY} \bbE_{\bx|\bY;\Theta}(\bY_{jl}\hbu_{jl, \bY, \Theta}^T) \|_F < c_2 \epsilon_2.
        \end{equation*}
        
        From Lemma \ref{lem:pertubation-sing}, we have
        \begin{equation*}
            \sigma_{\min}( \bbE_{\bY} \bbE_{\bx|\bY;\Theta}(\bY_{jl}\hbu_{jl, \bY, \Theta}^T))
             > \sigma_{\min}(\bB^*_j\bSigma^*_{jl}) - c_2 \epsilon_2 > \mathtt{b}_{\inf} > 0.
        \end{equation*}
        
        Let $\tr_b = \min_{j,l} \tr_{b,jl}$, then (4) is hold.
        
        Finally, we set $\tr_1 = \min(\tr_a, \tr_b)$ to obtain (1) to (4).
    \end{proof}

\subsection{Proof of Lemma \ref{lem:extend-omega-min}} \label{pro:extend-omega-min}
\begin{proof}
    Since $G(\bC, \Theta)$ is a continuous function of $\Theta$, $\forall \eta_1, \eta_2, \bC$, $\exists \tr_2$ that $\forall \Theta \in \mathbb{B}_2(\Theta^*, \tr_2)$, we have $|G(\bC,\Theta) - G(\bC, \Theta^*)| < \frac{1}{2} (\eta_1 - \eta_2)$, for some $0 < \eta_2 < \eta_1$.
    
    And from Condition \ref{ass:omega-min}, $\forall \pi \notin \Pi_0^*$, we have 
    \begin{equation}
        G(\bC^*, \Theta^*) - G(\bC_{\Theta^*}^*(\pi), \Theta^*) < - \eta_1.
    \end{equation}
    
    Then $\forall \pi \notin \Pi_0^*$, we have
    \begin{align*}
        G(\bC^*,\Theta) - G(\bC^*_{\Theta}(\pi), \Theta) & \leq |G(\bC^*,\Theta) - G(\bC^*,\Theta^*)| \\
        & \quad + G(\bC^*, \Theta^*) - G(\bC_{\Theta^*}^*(\pi), \Theta^*) \\
        & \quad + |G(\bC^*_{\Theta}(\pi), \Theta) - G(\bC^*_{\Theta}(\pi), \Theta^*)| \\
        & < \frac{1}{2}(\eta_1 - \eta_2) - \eta_1 - \frac{1}{2}(\eta_1 - \eta_2) \\
        & = - \eta_2.
    \end{align*}
    
    Therefore, $\forall \pi \notin \Pi_0^*$, we have:
    \begin{equation*}
        G(\bC^*_{\Theta},\Theta) - G(\bC^*_{\Theta}(\pi), \Theta) \leq G(\bC^*,\Theta) - G(\bC^*_{\Theta}(\pi), \Theta) < - \eta_2.
    \end{equation*}
    
    This shows that $\bC^*_{\Theta}(\pi)$ is not the minimum solution of $G(\bC, \Theta)$, and we simultaneously obtain (1) and (2). 
\end{proof}

\subsection{Proof of Lemma \ref{lem:error-bound}} \label{pro:error-bound}

\textbf{For Lemma \ref{lem:error-bound} (1)}:
\begin{proof}
For a fixed $\Theta \in \mathbb{B}_2(\Theta^*, \tr)$, let $\hat{\bC}$ be the estimator that minimizes $G_n(\bC, \Theta) + \lambda\|\bC\|_{l_1/F}$ and is consistent with causal order $\hat{\pi}$. We have
\begin{equation} \label{equ:main-theorem2.1}
\begin{split}
    & \quad \frac{1}{N} \bbE_{\bX|\calY;\Theta} \|\bX \bC^*_{\Theta}(\hat{\pi}) - \bX \hat{\bC}\|_F^2 + \lambda \|\hat{\bC}\|_{l_1/F} \\
    & \leq \frac{1}{N} \bbE_{\bX|\calY;\Theta} (\| \bX -\bX \bC^*_{\Theta} \|_F^2 - \| \bX -\bX \bC^*_{\Theta}(\hat{\pi}) \|_F^2) \\
    & \quad \quad + \frac{2}{N} \bbE_{\bX|\calY;\Theta} \langle \bX - \bX\bC^*_{\Theta}(\hat{\pi}), \bX(\hat{\bC} - \bC^*_{\Theta}(\hat{\pi})) \rangle + \lambda \|\bC^*_{\Theta}(\hat{\pi}) \|_{l_1/F}\\
    & \leq (I) + (II) + \lambda\|\bC^*_{\Theta}(\hat{\pi})\|_{l_1/F},
\end{split}
\end{equation}

where $\langle \cdot, \cdot \rangle$ denotes the inner product, and $\|\cdot\|_F$ denotes the Frobenius norm. Next, we will gives the upper bound for terms (I) and (II).

Bound (I):
\begin{align*}
    (I) &= G_n(\bC^*_{\Theta}, \Theta) - G_n(\bC^*_{\Theta}(\hat{\pi}), \Theta) \\
    &\leq | G_n(\bC^*_{\Theta}, \Theta) -  G(\bC^*_{\Theta}, \Theta)| + G(\bC^*_{\Theta}, \Theta) -  G(\bC^*_{\Theta}(\hat{\pi}), \Theta) + | G_n(\bC^*_{\Theta}(\hat{\pi}), \Theta) -  G(\bC^*_{\Theta}(\hat{\pi}), \Theta)|.
\end{align*}


We have the following statements, which show that the term $G_n(\bC, \Theta) - G(\bC, \Theta)$ has expectation 0 and bounded variance:

(1) $\bbE_\bY(G_n(\bC, \Theta) - G(\bC, \Theta)) = 0$;

(2) ${\rm Var} (G_n(\bC, \Theta) - G(\bC, \Theta)) = \frac{1}{N}{\rm Var}(\bbE_{\bx|\bY;\Theta}\| \bx - \bx \bC\|_F^2) \leq \frac{\|\bI-\bC\|_F^4 \mathtt{x}_{\sup}^4}{N}$.

By Chebyshev's inequality, we have:
\begin{align} \label{equ:bound1-chebyshev}
    P\left(|G_n(\bC, \Theta) - G(\bC, \Theta)|>\sqrt{\frac{\|\bI-\bC\|_F^4 \mathtt{x}_{\sup}^4}{\varrho_1 N}}\right) < \varrho_1
\end{align}

Using Eq. (\ref{equ:bound1-chebyshev}) in $(I)$, we obtain the following inequality with probability at least $1 - 2\varrho_1$:
\begin{equation}
    (I) \leq G(\bC^*_{\Theta}, \Theta) -  G(\bC^*_{\Theta}(\hat{\pi}), \Theta) + \sqrt{\frac{\|\bI-\bC_{\Theta}^*\|_F^4 \mathtt{x}_{\sup}^4}{\varrho_1 N}} + \sqrt{\frac{\|\bI-\bC_{\Theta}^*(\hat{\pi})\|_F^4 \mathtt{x}_{\sup}^4}{\varrho_1 N}}
\end{equation}

   Bound (II):

To bound the second term, we aim to show that the following equation holds true with high probability for $\delta_1 \in (0,1/2)$:
\begin{equation} \label{equ:boundII}
    \begin{split}
        & \quad \frac{1}{N} \bbE_{\bX|\calY;\Theta} \langle \bX - \bX\bC^*_{\Theta}(\hat{\pi}), \bX(\hat{\bC} - \bC^*_{\Theta}(\hat{\pi})) \rangle \\
        & \leq \frac{\delta_1}{2N} \bbE_{\bX|\calY;\Theta} \| \bX (\hat{\bC} - \bC^*_{\Theta}(\hat{\pi})) \|_F^2 + \delta_1\lambda \|\hat{\bC}-\bC^*_{\Theta}(\hat{\pi}) \|_{l_1/F}
    \end{split}
\end{equation}



Let $\be_j(\pi) \in \mathbb{R}^N$ as the $j$-th column of matrix $\bX - \bX\bC^*_{\Theta}(\hat{\pi})$ and $\bbeta \in \mathbb{R}^M$ as the $j$-th column of matrix $\hat{\bC} - \bC^*_{\Theta}(\hat{\pi})$. Denote $\calE_j$ is the event:
\begin{equation} \label{equ:event}
    \calE_j := \left\{ \underset{\bbeta \in \mathbb{R}^M}{\sup} \frac{1}{N} \bbE_{\bX|\calY;\Theta} \langle \be_j(\hat{\pi}), \bX \bbeta \rangle - \frac{\delta_1}{2N} \bbE_{\bX|\calY;\Theta} \|\bX \bbeta \|_2^2 - \delta_1 \lambda \|\bbeta \|_{l_1/l_2} \leq 0 \right\},
\end{equation}

where $\bbeta = [\bbeta_1, \bbeta_2, \ldots, \bbeta_P]$ for $\bbeta_{i} \in \mathbb{R}^{L_iK_i}$ and $\|\bbeta \|_{l_1/l_2} = \sum_{i} \|\bbeta_i\|_2$. 

Therefore, to prove Eq. (\ref{equ:boundII}), it suffices to show that for any given column $j$ and causal order $\hat{\pi}$, the event $\calE$ hold with a high probability.

We can then express $\calE_j$ as:
\begin{equation*}
    \begin{split}
        \calE_j &\subseteq \left\{ \underset{\bbeta \in \mathbb{R}^M}{\sup} \frac{1}{2N} \bbE_{\bX|\calY;\Theta} \| \frac{\be_j(\hat{\pi})}{\delta_1} \|_2^2 - \frac{1}{2N} \bbE_{\bX|\calY;\Theta} \| \frac{\be_j(\hat{\pi})}{\delta_1} - \bX \bbeta  \|_2^2 +  \lambda \|\bbeta \|_{l_1/l_2} \leq 0  \right\} \\
        &= \left\{ \textbf{0} \in \underset{\bbeta \in \mathbb{R}^M}{\arg \min} \frac{1}{2N}  \bbE_{\bX|\calY;\Theta} \| \frac{\be_j(\hat{\pi})}{\delta_1} - \bX \bbeta \|_2^2 + \lambda \| \bbeta \|_{l_1/l_2}  \right\}
    \end{split}
\end{equation*}

Event $\calE_j$ is correspond to the Null-consistency of group lasso problem, we use Lemma \ref{lem:null-consistency} to find the solution $\bbeta$ and $\bw$,
\begin{align*}
    \bbeta &= \textbf{0}, \\
    \bw &= \frac{1}{\lambda N} \bbE_{\bX|\calY;\Theta}(\bX^T \frac{\be_j(\hat{\pi})}{\delta_1}).
\end{align*}

Next we proof that $\| \bw \|_{l_\infty/l_2} \leq 1$ holds with a high probability, where $\| \bw \|_{l_\infty/l_2} = \max_{i=1,\ldots,P} \|\bw_{i}\|_2$ and $\bw_{i}$ is the gradient corresponds to $\bbeta_{i}$. To proof this, we bound the variance of $\|\bw\|_2$. 

We first prove that the expectation of $\bw$ is \textbf{0} from Lemma \ref{lem:expectation-0}, and we have 
\begin{align*}
    \|\bw\|_2^2 &\leq \frac{1}{\lambda^2 N^2 \delta_1^2} \bbE_{\bX|\calY,\Theta}\| \sum_{n=1}^N \bx^{(n) T} \be^{(n)}_j(\hat{\pi})\|_2^2 \\
    &= \frac{1}{\lambda^2 N^2 \delta_1^2} \sum_{n=1}^N \bbE_{\bx^{(n)}|\bY^{(n)},\Theta}\| \bx^{(n) T} \be^{(n)}_j(\hat{\pi})\|_2^2 \\
    &\leq \frac{1}{\lambda^2 N^2 \delta_1^2} \sum_{n=1}^N \bbE_{\bx^{(n)}|\bY^{(n)},\Theta}\| \bx^{(n) T} \bx^{(n)} (\bI - \bC^*_{\Theta}(\pi))\|_F^2 \\
    &\leq \frac{\| \bI - \bC^*_{\Theta}(\pi) \|_F^2}{\lambda^2 N^2 \delta_1^2} \sum_{n=1}^N \bbE_{\bx^{(n)}|\bY^{(n)},\Theta}\| \bx^{(n)}\|_2^4,
\end{align*}

where,
\begin{align*}
    \bbE_\bY \left( \sum_{n=1}^N \bbE_{\bx^{(n)}|\bY^{(n)},\Theta}\| \bx^{(n)}\|_2^4 \right) &\leq N \mathtt{x}_{\sup}^4 , \\
    {\rm Var} \left( \sum_{n=1}^N \bbE_{\bx^{(n)}|\bY^{(n)},\Theta}\| \bx^{(n)}\|_2^4 \right) &\leq N \mathtt{x}_{\sup}^8 .
\end{align*}

Suppose we have $\frac{\| \bI - \bC^*_{\Theta}(\pi) \|_F^2 \mathtt{x}_{\sup}^4}{\lambda^2 N \delta_1^2} < 1$. By Chebyshev's inequality, we have
\begin{equation*}
    P(\|\bw\|_2^2 \geq 1) \leq \frac{\frac{\| \bI - \bC^*_{\Theta}(\pi) \|_F^4}{\lambda^4 N^3 \delta_1^4} \bx_{\sup}^8}{\left(1 - \frac{\| \bI - \bC^*_{\Theta}(\pi) \|_F^2 \mathtt{x}_{\sup}^4}{\lambda^2 N \delta_1^2}\right)^2} \leq \frac{\frac{\mathtt{d}_{\sup}^4}{\lambda^4 N^3 \delta_1^4} \bx_{\sup}^8}{\left(1 - \frac{\mathtt{d}_{\sup}^2 \mathtt{x}_{\sup}^4}{\lambda^2 N \delta_1^2}\right)^2} := \varrho_2
\end{equation*}


Since $\|\bw\|_{l_\infty/l_2} \leq \| \bw \|_2$, we have:
\begin{equation*}
    P(\| \bw \|_{l_\infty/l_2} \geq 1) \leq \varrho_2.
\end{equation*}

Thus, with probability $1 - \varrho_2$, event $\calE_j$ holds true. Taking uniform control over all possible $j=1,2,\ldots,M$ and $\hat{\pi}$, we conclude that with probability $1 - MP! \varrho_2$, Eq. (\ref{equ:boundII}) holds true.

Finally, for Lemma \ref{lem:error-bound}(1), suppose $\hat{\pi} \notin \Pi_0^*$, then $ G(\bC^*_{\Theta}, \Theta) -  G(\bC^*_{\Theta}(\hat{\pi}), \Theta) < -\eta_2$, and we back to Eq. (\ref{equ:main-theorem2.1}). With probability $1-\varrho_1-MP!\varrho_2$, we have:
\begin{equation} 
\begin{split}
    & \quad \frac{1}{N} \bbE_{\bX|\calY;\Theta} \|\bX \bC^*_{\Theta}(\hat{\pi}) - \bX \hat{\bC}\|_F^2 + \lambda \|\hat{\bC}\|_{l_1/F} \\
    & \leq -\eta_2 + \sqrt{\frac{\|\bI-\bC_{\Theta}^*\|_F^4 \mathtt{x}_{\sup}^4}{\varrho_1 N}} + \sqrt{\frac{\|\bI-\bC_{\Theta}^*(\hat{\pi})\|_F^4 \mathtt{x}_{\sup}^4}{\varrho_1 N}} \\
    & \quad + \frac{\delta_1}{N} \bbE_{\bX|\calY;\Theta} \| \bX (\hat{\bC} - \bC^*_{\Theta}(\hat{\pi})) \|_F^2 + 2\delta_1\lambda \|\hat{\bC}-\bC^*_{\Theta}(\hat{\pi}) \|_{l_1/F} + \lambda \|\bC^*_{\Theta}(\hat{\pi})\|_{l_1/F}.
\end{split}
\end{equation}

For $\delta_1 \in (0,1)$, we have:
\begin{equation*}
\begin{split}
    & \quad \frac{1}{N} \bbE_{\bX|\calY;\Theta} \|\bX \bC^*_{\Theta}(\hat{\pi}) - \bX \hat{\bC}\|_F^2 \\
    & \leq -\eta_2 + \sqrt{\frac{\|\bI-\bC_{\Theta}^*\|_F^4 \mathtt{x}_{\sup}^4}{\varrho_1 N}} + \sqrt{\frac{\|\bI-\bC_{\Theta}^*(\hat{\pi})\|_F^4 \mathtt{x}_{\sup}^4}{\varrho_1 N}} + \lambda(2\delta_1 + 1)  \mathtt{c}_{\sup}.\\
\end{split}
\end{equation*}

It contradicts with the condition that:
\begin{equation*}
    \eta_2 > 2\sqrt{\frac{\mathtt{d}_{\sup}^4 \mathtt{x}_{\sup}^4}{\varrho_1 N}} + \lambda(2\delta_1 + 1)  \mathtt{c}_{\sup}.
\end{equation*}

\textbf{For Lemma \ref{lem:error-bound}(2)}: we denote that $\bDelta = \hat{\bC} - \bC^*_{\Theta}(\hat{\pi})$, we have:
\begin{align*}
    \frac{1}{N}  \bbE_{\bX|\calY;\Theta} \| \bX \bDelta\|_F^2 &= \frac{1}{N} \sum_{n=1}^N \|\hbu_{\Theta, \bY^{(n)}} \bDelta \|_F^2  + {\rm tr}(\bDelta^T \hat{\bSigma}_{\Theta} \bDelta) \\
    &= {\rm tr}(\bDelta^T(\frac{1}{N} \sum_{n=1}^N  \hbu^T_{\Theta, \bY^{(n)}}\hbu_{\Theta, \bY^{(n)}})\bDelta)  + {\rm tr}(\bDelta^T \hat{\bSigma}_{\Theta} \bDelta) \\ 
    &\geq \|\bDelta\|_F^2 \text{\rm minEig}\left(\frac{1}{N} \sum_{n=1}^N  \hbu^T_{\Theta, \bY^{(n)}}\hbu_{\Theta, \bY^{(n)}} + \hat{\bSigma}_{\Theta}\right).
\end{align*}

Denote $\bPhi_{\Theta}:=\frac{1}{N} \sum_{n=1}^N \hbu^T_{\Theta, \bY^{(n)}}\hbu_{\Theta, \bY^{(n)}} + \hat{\bSigma}_{\Theta}$ and denote $\bar{\bPhi}_{\Theta} := \bbE_\calY(\bPhi_{\Theta})$. From Lemma \ref{lem:extend-bound}, we have $\text{\rm minEig}(\bar{\bPhi}_{\Theta}) > \mathtt{s}_{\inf}$, and 
\begin{align*}
    \| \bPhi_{\Theta} - \bar{\bPhi}_{\Theta} \|_F^2 &= \|\frac{1}{N} \sum_{n=1}^N \hbu^T_{\Theta, \bY^{(n)}}\hbu_{\Theta, \bY^{(n)}} - \bar{\bPhi}_{\Theta} \|_F^2 \\
    &= \frac{1}{N^2} \sum_{n=1}^N \| \hbu^T_{\Theta, \bY^{(n)}}\hbu_{\Theta, \bY^{(n)}} - \bar{\bPhi}_{\Theta} \|_F^2 \\
    &= \frac{1}{N^2} \sum_{n=1}^N \| \hbu^T_{\Theta, \bY^{(n)}}\hbu_{\Theta, \bY^{(n)}} - \sum_{n'=1}^N  (\hbu^T_{\Theta, \bY^{(n')}}\hbu_{\Theta, \bY^{(n')}}) \|_F^2. \\
\end{align*}

where,
\begin{align*}
    \bbE\left(\sum_{n=1}^N \| \hbu^T_{\Theta, \bY^{(n)}}\hbu_{\Theta, \bY^{(n)}} - \sum_{n'=1}^N  (\hbu^T_{\Theta, \bY^{(n')}}\hbu_{\Theta, \bY^{(n')}}) \|_F^2\right) \leq N \bx_{\sup}^4, \\
    {\rm Var}\left(\sum_{n=1}^N \| \hbu^T_{\Theta, \bY^{(n)}}\hbu_{\Theta, \bY^{(n)}} - \sum_{n'=1}^N  (\hbu^T_{\Theta, \bY^{(n')}}\hbu_{\Theta, \bY^{(n')}}) \|_F^2\right) \leq N \bx_{\sup}^8.
\end{align*}

By Chebyshev’s inequality, we have:
\begin{align*}
    P\left(\| \bPhi_{\Theta} - \bar{\bPhi}_{\Theta} \|_F^2 \geq \frac{\mathtt{x}_{\sup}^4}{N} + \sqrt{\frac{\mathtt{x}_{\sup}^8}{\varrho_3 N}} \right) &< \varrho_3, \\
    P\left(|\text{\rm minEig}(\bPhi_{\Theta}) - \text{\rm minEig}(\bar{\bPhi}_{\Theta})| \geq \sqrt{\frac{\mathtt{x}_{\sup}^4}{N} + \sqrt{\frac{\mathtt{x}_{\sup}^8}{\varrho_3 N}}}\right) &< \varrho_3, \\
    P\left(\text{\rm minEig}(\bPhi_{\Theta}) \geq \mathtt{s}_{\inf} - \sqrt{\frac{\mathtt{x}_{\sup}^4}{N} + \sqrt{\frac{\mathtt{x}_{\sup}^8}{\varrho_3 N}}}\right) > 1 - \varrho_3. \\
\end{align*}

Then, with at least probability $1- 2\varrho_1 - P!M\varrho_2 - \varrho_3$, we have $\hat{\pi} \in \Pi_0^*$, therefore:
\begin{align*}
    \|\bDelta\|_F^2 &\leq \frac{\frac{1}{N}  \bbE_{\bX|\calY;\Theta} \| \bX \bDelta\|_F^2}{\mathtt{s}_{\inf} - \sqrt{\frac{\mathtt{x}_{\sup}^4}{N} + \sqrt{\frac{\mathtt{x}_{\sup}^8}{\varrho_3 N}}}} \\
    &\leq \frac{2 \sqrt{ \frac{\mathtt{d}_{\sup}^4 \mathtt{x}_{\sup}^4}{\varrho_1 N}} + \lambda(2\delta_1 + 1)  \mathtt{c}_{\sup}}{\mathtt{s}_{\inf} - \sqrt{\frac{\mathtt{x}_{\sup}^4}{N} + \sqrt{\frac{\mathtt{x}_{\sup}^8}{\varrho_3 N}}}}.
\end{align*}
\end{proof}

\subsection{Proof of Lemma \ref{lem:polar-decomposition}} \label{pro:polar-decomposition}
\begin{proof}
    Because $\hat{\bB}_j^T \hat{\bB}_j = \bI$, ${\rm tr}(\hat{\bB}_j^T\hat{\bSigma}_{jl}\hat{\bB}_j) = {\rm tr}(\hat{\bSigma}_{jl})$ is a constant unrelated to $\hat{\bB}_j$. For a fixed $j$, the estimator of $\hat{\bB}_j$ is given by:
    \begin{equation*}
    \begin{split}
        & \hat{\bB}_j = \underset{\bB_j}{\arg \min} \frac{1}{NL_j} \sum_{n=1}^N \sum_{l=1}^{L_j} \| \bY^{(n)}_{jl} \hbu^{(n) T}_{jl; \Theta, \bY^{(n)}} - \bB_j \|_F^2 \\
        & {\rm s.t.} \quad \bB_j^T \bB_j = \bI.
    \end{split}
    \end{equation*}
     We denote $\bZ = \frac{1}{NL_j} \sum_{n=1}^N \sum_{l=1}^{L_j} \bY^{(n)}_{jl} \hbu^{(n) T}_{jl;
     \Theta, \bY^{(n)}}$ and $\bar{\bZ} = \bbE_{\bY}(\bZ)$. We consider $\bZ$ is a small perturbation of $\bZ = \bar{\bZ} + \bE$ and use the perturbation theory of Polar decomposition. From \cite{li1993perturbation}, we obtain that:
     \begin{equation} \label{equ:bound-Bj}
         \| \hat{\bB}_j - \bB_{\Theta j}^*\|_F \leq \frac{\|\bZ-\bar{\bZ}\|_F}{\min\{\|\bZ^+\|_2^{-1}, \|\bar{\bZ}^+\|_2^{-1}\}},
     \end{equation}

     where $\|\bZ^+\|_2^{-1}$ and $\|\bar{\bZ}^+\|_2^{-1}$ is smallest singular value of $\bZ$ and $\bar{\bZ}$ greater than 0. Next, we bound the numerator and denominator of RHS of Eq. (\ref{equ:bound-Bj}).

     For the numerator, we have 
     \begin{align*}
         \|\bE\|_F^2 &= \| \frac{1}{NL_j} \sum_{n=1}^N \sum_{l=1}^{L_j} \bY^{(n)}_{jl} \hbu^{(n) T}_{jl;
     \Theta, \bY^{(n)}} - \bar{\bZ} \|_F^2 \\
        &= \frac{1}{N^2} \sum_{n=1}^N \| \frac{1}{L_j} \sum_{l=1}^{L_j}( \bY^{(n)}_{jl} \hbu^{(n) T}_{jl; \Theta, \bY^{(n)}}) - \bZ \|_F^2,
     \end{align*} 
     where
     \begin{align*}
         & \bbE_\calY\left(\sum_{n=1}^N \| \frac{1}{L_j} \sum_{l=1}^{L_j}( \bY^{(n)}_{jl} \hbu^{(n) T}_{jl; \Theta, \bY^{(n)}}) - \bZ \|_F^2\right) \\
         = & N \bbE_\bY\left(\| \frac{1}{L_j} \sum_{l=1}^{L_j}( \bY_{jl} \hbu^{T}_{jl; \Theta, \bY}) - \bZ \|_F^2\right) \\
         = & \frac{N}{L_j^2} \bbE_\bY\left(\| \sum_{l=1}^{L_j}( \bY_{jl} \hbu^{T}_{jl; \Theta, \bY} - \bZ) \|_F^2\right)\\
         \leq & \frac{N}{L_j} \bbE_\bY\left(\sum_{l=1}^{L_j}\| \bY_{jl} \hbu^{T}_{jl; \Theta, \bY} \|_F^2\right)\\
         \leq & N \mathtt{y}_{\sup}^2\\
    \end{align*}
    and 
    \begin{align*}
         & {\rm Var}\left(\sum_{n=1}^N \| \frac{1}{L_j} \sum_{l=1}^{L_j}( \bY^{(n)}_{jl} \hbu^{(n) T}_{jl; \Theta, \bY^{(n)}}) - \bZ \|_F^2\right) \leq N \mathtt{y}_{\sup}^4.
     \end{align*}
     
     then by Chebyshev's inequality, we have:
     \begin{equation} \label{equ:bound-Bj-numerator}
         P\left(\|\bZ-\bar{\bZ}\|_F^2 \geq \frac{\mathtt{y}^2_{\sup}}{N} + \sqrt{\frac{ \mathtt{y}^4_{\sup}}{\varrho_4 N}} \right) < \varrho_4.
     \end{equation}
     
     For the denominator, from Lemma \ref{lem:extend-bound}(3). By Chebyshev's inequality, we have:
     \begin{equation} \label{equ:bound-Bj-denominator}
         P\left(\|\bZ^+\|_2^{-1} \leq \mathtt{b}_{\inf} - \sqrt{\frac{\mathtt{y}^2_{\sup}}{N} + \sqrt{\frac{ \mathtt{y}^4_{\sup}}{\varrho_5 N}}}\right) < \varrho_5.
     \end{equation}

     Combine Eq. (\ref{equ:bound-Bj-numerator}) and Eq. (\ref{equ:bound-Bj-denominator}), at least probability $1-\varrho_4-\varrho_5$, we have:
     \begin{equation} 
         \| \hat{\bB}_j - \bB_{\Theta j}^*\|_F^2 \leq \frac{\frac{\mathtt{y}^2_{\sup}}{N} + \sqrt{\frac{ \mathtt{y}^4_{\sup}}{\varrho_4 N}} }{\left(\mathtt{b}_{\inf} - \sqrt{\frac{\mathtt{y}^2_{\sup}}{N} + \sqrt{\frac{ \mathtt{y}^4_{\sup}}{\varrho_5 N}}}\right)^2}.
     \end{equation}

     Finally, we take the uniform control for all nodes $j=1,2,\ldots,P$, then with probability $1 - P\varrho_4 - P\varrho_5$, we have:
    \begin{equation*}
        \|\hat{\bB} - \bB^*_{\Theta}\|_F^2 \leq \frac{P\left(\frac{\mathtt{y}^2_{\sup}}{N} + \sqrt{\frac{ \mathtt{y}^4_{\sup}}{\varrho_4 N}} \right) }{\left(\mathtt{b}_{\inf} - \sqrt{\frac{\mathtt{y}^2_{\sup}}{N} + \sqrt{\frac{ \mathtt{y}^4_{\sup}}{\varrho_5 N}}}\right)^2}.
    \end{equation*}
\end{proof}

\section{Minor Lemma and Derivation}
\subsection{Minor Lemma}
\begin{lemma} \label{lem:post-true}
    $\mathbb{E}_{\bY}(\hbu_{\Theta^*, \bY}) = \mathbf{0}$ and ${\rm Cov}(\hbu_{\Theta^*, \bY}) + \hat{\bSigma}_{\Theta^*} = (\bI - \bC^*)^{-T}\omega_0^{2 *}(\bI - \bC^*)^{-1}$. 
\end{lemma}
\begin{proof}
    We have 
    \begin{align*}
        \mathbb{E}_{\bY}\mathbb{E}_{\bx|\bY;\Theta^*} (\bx) &= \mathbb{E}_{\bx|\Theta^*}(\bx) = \mathbf{0}, \\ 
        {\rm Cov}\mathbb{E}_{\bx|\bY;\Theta^*} (\bx) + \mathbb{E}_{\bY} {\rm Cov}_{\bx|\bY;\Theta^*}(\bx) &= {\rm Cov}_{\bx|\Theta^*}(\bx) = (\bI - \bC^*)^{-T}\omega_0^{2 *}(\bI - \bC^*)^{-1}.
    \end{align*}
\end{proof}
\begin{lemma} \label{lem:pertubation-eig}
    For any positive definite matrix $\mathbf{A}$ and $\mathbf{B}, {\rm minEig}(\mathbf{A}) - {\rm minEig}(\mathbf{B}) \leq \| \mathbf{A} - \mathbf{B}\|_F$.
\end{lemma}
\begin{proof}
    It is straightforward from \citet{li1994relative} that we have
    \begin{equation*}
        \sqrt{\sum_{i} (\lambda_{A,i} - \lambda_{B,r(i)})^2 }\leq \|\mathbf{A} - \mathbf{B}\|_F,
    \end{equation*}
    where $\lambda_A$ and $\lambda_B$ are the eigenvalues of matrix $\mathbf{A}$ and $\mathbf{B}$.
\end{proof}
\begin{lemma} \label{lem:pertubation-sing}
     For any matrix $\mathbf{A}$ and $\mathbf{B}$, we have $\sigma_{\min}(\mathbf{A}) - \sigma_{\min}(\mathbf{B}) \leq \| \mathbf{A} - \mathbf{B}\|_F$
\end{lemma}
\begin{proof}
     It is straightforward from \cite{mirsky1960symmetric} that we have:
    \begin{equation*}
        \sqrt{\sum_{i} (\sigma_{A,i} - \sigma_{B,i})^2 }\leq \|\mathbf{A} - \mathbf{B}\|_F,
    \end{equation*}
    where $\sigma_A$ and $\sigma_B$ are the singular values of matrix $\mathbf{A}$ and $\mathbf{B}$.
\end{proof}
\begin{lemma} \label{lem:null-consistency}
    The $\bbeta = 0$ is the optimal solution of the penalized Lasso with $l_1/l_2$ penalty $\frac{1}{2N} \bbE_{\bX|\calY;\Theta }\|\frac{\be_j(\hat{\pi})}{\delta_1} - \bX \bbeta \| + \lambda \| \bbeta \|_{l_1/F} $ if the following condition is hold:
    \begin{align*}
        \bw \in \partial \| \bbeta \|_{l_1/F}, \\
        \frac{1}{2N} \mathbb{E}_{\bX|\calY; \Theta} \bX^T(\frac{\be_j(\hat{\pi})}{\delta_1} - \bX \bbeta) + \lambda \bw = 0, \\
        \|\bw\|_{l_\infty/l_2} < 1.
    \end{align*}
\end{lemma}
\begin{proof}
    It is straightforward by following Lemma 1 in \cite{aragam2015learning}.
\end{proof}

\begin{lemma}\label{lem:expectation-0}
    $\forall \Theta$ and $\pi$, denote $\be_j(\pi)$ is the $j$-th column of $\bX - \bX \bC^*_{\Theta}(\pi)$. We have $\bbE_{\calY} \bbE_{\bX|\calY, \Theta}(\bX^T \be_j(\pi)) = \mathbf{0}, \forall j$.
\end{lemma}
\begin{proof}
     Since $\be_j(\hat{\pi})$ is the $j$-th column of $\bX - \bX \bC^*_{\Theta}(\pi)$. Therefore, $\bC^*_{\Theta}$ satisfies
     \begin{align*}
         \frac{\partial G(\bC^*_{\Theta}, \Theta)}{\partial \bC} &= 0, \\
         \bbE_{\bY}\bbE_{\bx|\bY;\Theta}(\bx(\bx_j - \bx \bC_{\Theta}^*(\hat{\pi})_j )) &= \mathbf{0}, \forall j, \\
         \bbE_{\bY}\bbE_{\bx|\bY;\Theta}(\bx \be_j(\pi))) &= \mathbf{0}. 
     \end{align*} 
     
     Therefore, we have $\bbE_\bY \bbE_{\bx|\bY,\Theta} (\bX^T \frac{\be_j(\pi)}{\delta_1}) = \mathbf{0}$.
\end{proof}

\begin{lemma}[\citet{balakrishnan2017statistical}] \label{lem:contractive} For radius $\tr>0$ and pair $(\gamma, \beta)$ satisfying $0\leq \gamma < \beta$, suppose that the function $Q(\cdot,\Theta^*)$ is globally $\beta$-strongly concave, and the Condition \ref{ass:Lipschitz-Gradient} holds on the ball $\mathbb{B}_2(\Theta^*,\tr)$. Then the EM operator is contractive over $\mathbb{B}_2(\Theta^*,\tr)$, in particular with:
\begin{equation*}
    D(M(\Theta), \Theta^*) \leq \frac{\gamma}{\beta} D(\Theta, \Theta^*)
\end{equation*}
\end{lemma}

\begin{lemma}\label{lem:bound-r}
    Denote the $\hat{\br}_{jl,\Theta}^{2*}$ as the variance determined by $\bB_{\Theta}^*$ and $\hat{\br}$ as the variance determined by $\hat{\bB}$ from Eq. (\ref{equ:M-step-u2}). Under  Lemma \ref{lem:polar-decomposition}, we have
    \begin{equation*}
        \| \hat{\br}^2 - \hat{\br}_{jl,\Theta}^{2*} \|_2^2 \leq \frac{M\mathtt{y}^4_{\sup}}{N\varrho_6} + O(\frac{1}{N^2})
    \end{equation*}
    with probability $1-M\varrho_6$. 
\end{lemma}
\begin{proof}
    We have
    \begin{align*}
        \hat{r}_{jl}^2 &= \frac{1}{N} \sum_{n=1}^N \bbE_{\bx^{(n)}|\bY^{(n)},\Theta} \| \bY^{(n)}_{jl} - \hat{\bB} \bx_{jl}^{(n)} \|_2^2 \\
        &= \frac{1}{N} \sum_{n=1}^N \bbE_{\bx^{(n)}|\bY^{(n)},\Theta} \| \bY^{(n)}_{jl} - \hat{\bB}^*_{\Theta} \bx_{jl}^{(n)} +  \hat{\bB}^*_{\Theta} \bx_{jl}^{(n)} - \hat{\bB} \bx_{jl}^{(n)} \|_2^2 \\
        &= \frac{1}{N} \sum_{n=1}^N \bbE_{\bx^{(n)}|\bY^{(n)},\Theta}  (\| \bY^{(n)}_{jl} - \hat{\bB}^*_{\Theta} \bx_{jl}^{(n)} \|_2^2 + \|  \hat{\bB}^*_{\Theta} \bx_{jl}^{(n)} - \hat{\bB} \bx_{jl}^{(n)} \|_2^2 + 2\langle \bY^{(n)}_{jl} - \hat{\bB}^*_{\Theta} \bx_{jl}^{(n)} , \hat{\bB}^*_{\Theta} \bx_{jl}^{(n)} - \hat{\bB} \bx_{jl}^{(n)} \rangle).
    \end{align*}

    Denote $e_{r} = \hat{r}_{jl}^2 - \frac{1}{N} \sum_{n=1}^N \bbE_{\bx_{jl}^{(n)}|\bY^{(n)},\Theta}  \| \bY^{(n)}_{jl} - \hat{\bB}^*_{\Theta} \bx_{jl}^{(n)} \|_2^2$, we have
    \begin{align*}
    e_r &= \frac{1}{N} \sum_{n=1}^N \bbE_{\bx_{jl}^{(n)}|\bY^{(n)},\Theta}( \|\hat{\bB}^*_{\Theta} \bx_{jl}^{(n)} - \hat{\bB} \bx_{jl}^{(n)} \|_2^2 + 2\langle \bY^{(n)}_{jl} - \hat{\bB}^*_{\Theta} \bx_{jl}^{(n)} , \hat{\bB}^*_{\Theta} \bx_{jl}^{(n)} - \hat{\bB} \bx_{jl}^{(n)} \rangle) \\
    & \leq \frac{\delta_B}{N} (\mathtt{x}_{\sup}^2 + 2 \| \bB_\Theta^* \|_F \mathtt{x}_{\sup}^2 + 2 \mathtt{y}_{\sup}).
    \end{align*}

    Denote $e'_r =  \hat{r}^{*2}_{jl, \Theta} - \frac{1}{N} \sum_{n=1}^N \bbE_{\bx_{jl}^{(n)}|\bY^{(n)},\Theta}  \| \bY^{(n)}_{jl} - \hat{\bB}^*_{\Theta} \bx_{jl}^{(n)} \|_2^2$, we have
    \begin{align*}
        \bbE(e_r') &= 0, \\
        {\rm Var}(e_r') &\leq \frac{1}{N} \mathtt{y}_{\sup}^4.
    \end{align*}

    Therefore, using Chebyshev's inequality, with probability $\varrho_6$, we have
    \begin{equation*}
        P\left(e_r'>\sqrt{\frac{\mathtt{y}^4_{\sup}}{N\varrho_6}}\right) < \varrho_6.
    \end{equation*}
    and with $1-\varrho_6$, we have
    \begin{equation*}
        |e_r|+|e_r'| \leq \frac{\delta_B}{N} (\mathtt{x}_{\sup}^2 + 2 \| \bB_\Theta^* \|_F \mathtt{x}_{\sup}^2 + 2 \mathtt{y}_{\sup}) + \sqrt{\frac{\mathtt{y}^4_{\sup}}{N\varrho_6}}. 
    \end{equation*}

    Thus, 
    \begin{equation*}
        (\hat{r}_{jl}^2 - \hat{r}^{*2}_{jl, \Theta})^2 \leq\frac{\mathtt{y}^4_{\sup}}{N\varrho_6} + O(\frac{1}{N^2}) 
    \end{equation*}

    Taking uniform control of all $j,l$ that, with probability $1-M\varrho_6$ we have
    \begin{equation*}
        \| \hat{\br}^2 - \hat{\br}_{jl,\Theta}^{2*} \|_2^2 \leq \frac{M\mathtt{y}^4_{\sup}}{N\varrho_6} + O(\frac{1}{N^2})
    \end{equation*}
\end{proof}

\begin{lemma}\label{lem:bound-w}
    Denote the $\hat{\omega}_0^{*2}$ is the variance determined by $\bC_{\Theta}^*$ and $\hat{\omega}_0^2$ is the variance determined by $\hat{\bC}$ from Eq. (\ref{equ:m-step-c3}). Under Lemma \ref{lem:error-bound}, we have
    \begin{equation*}
        \hat{\omega}_0^{*2} - \hat{\omega}_0^2 \leq \sqrt{ \frac{\mathtt{d}_{\sup}^4 \mathtt{x}_{\sup}^4}{\varrho_1 N}} + \lambda(2\delta_1 + 1)  \mathtt{c}_{\sup}
    \end{equation*}
    with probability $1$.
\end{lemma}

\begin{proof}
    From Eq. (\ref{equ:m-step-c3}), we have
    \begin{align*}
        \omega_0^2 &= \frac{1}{NM} \sum_{n=1}^N \bbE_{\bx^{(n)}|\bY^{(n)}}\|\bx^{(n)} - \bx^{(n)} \hat{\bC}\|_2^2 \\
        & \leq 2 \sqrt{ \frac{\mathtt{d}_{\sup}^4 \mathtt{x}_{\sup}^4}{\varrho_1 N}} + \lambda(2\delta_1 + 1)  \mathtt{c}_{\sup}
    \end{align*}
\end{proof}

\subsection{Conditions to ensure the convergence of EM algorithm} \label{app:converge}
To utilize the theorem proposed by \citet{wang2015high} and \citet{balakrishnan2017statistical}, we denote $Q$ as the population analog of $Q_n$. Condition \ref{ass:concavity-smoothness} and \ref{ass:Lipschitz-Gradient} are common conditions to satisfy the convergence of EM algorithm. 
\begin{equation*}
       \begin{split}
       & Q(\Theta;\Theta') = \mathbb{E}_\bY \mathbb{E}_{\bx|\bY;\Theta'} \log f(\bx,\bY;\Theta) \\
       & = \quad \int p(\mathbf{Y};\Theta^*) \int p(\bx|\mathbf{Y};\Theta') \log f(\bx,\mathbf{Y};\Theta) {\rm d}\bx{\rm d}\mathbf{Y}.
       \end{split}
\end{equation*}

\begin{condition} 
    [\textbf{Concavity-Smoothness}]\label{ass:concavity-smoothness}
    For any $\Theta_1,\Theta_2 \in \mathbb{B}_2(\Theta^*, \tr)$, $Q(\cdot\ ;\Theta^*)$ is $\alpha$-smooth, i.e., denote the $\theta_1, \theta_2$ are the vector form of parameter set $\Theta_1, \Theta_2$, we have
    \begin{equation*}
        Q(\Theta_1,\Theta^*) \geq Q(\Theta_2,\Theta^*) + (\theta_1-\theta_2)^T \triangledown Q(\Theta_2;\Theta^*) - \frac{\alpha}{2}\|\theta_2-\theta_1\|_2,
    \end{equation*}
    
    and $\beta$-strongly concave, i.e.,
    \begin{equation*}
        Q(\Theta_1,\Theta^*) \leq Q(\Theta_2,\Theta^*) + (\theta_1-\theta_2)^T \triangledown Q(\Theta_2;\Theta^*) - \frac{\beta}{2}\|\theta_2-\theta_1\|_2.
    \end{equation*}
\end{condition}

\begin{condition}[\textbf{Lipschitz-Gradient}] \label{ass:Lipschitz-Gradient}
    For the true parameter $\Theta^*$ and any $\Theta \in \mathbb{B}_2(\Theta^*,r)$, denote $\theta, \theta^*$ are the vector form of parameter set $\Theta, \Theta^*$, we have:
    \begin{equation}
        \|\triangledown Q(M(\Theta);\Theta^*) - \triangledown Q(M(\Theta);\Theta)\|_2 \leq \gamma \| \theta - \theta^* \|_2
    \end{equation}
\end{condition}

\subsection{Computing Expectation} \label{app:Expectation}
\subsubsection{Forward filtering}
When using forward filtering in DAG, we need to know source of the noise, this process is implement by the matrix $\bG$ and $\bH$, which record the coefficients of the noise from Eq. (\ref{equ:matrix-bilinear-trans}) and Eq. (\ref{equ:Y-likelihood}).

We, denote: 
\begin{itemize}
    \item $\bX:$ $\bX=[\bX_1,\ldots,\bX_P]$ with size $N \times \sum{L_jK_j}$, which is the distribution of $\bX$ before forward filtering.
    \item $\Tilde{\bX}:$ $\tilde{\bX}=[\tilde{\bX}_1,\ldots,\tilde{\bX}_P]$ with size $N \times \sum{L_jK_j}$, which is the distribution of $\bX$ after forward filtering.
    \item $\hat{\bX}:$ $\hat{\bX}=[\hat{\bX}_1,\ldots,\hat{\bX}_P]$ with size $N \times \sum{L_jK_j}$, which is the distribution of $\bX$ after backward smoothing.
    \item $\bxi:$ $\bxi=[\bxi_1,\ldots,\bxi_P]$ with size $N \times \sum{L_jK_j}$, which is the noise from Eq. (\ref{equ:matrix-bilinear-trans}).
    \item $\bvarepsilon:$ $\bvarepsilon=[\varepsilon_{11}(t_1),\varepsilon_{11}(t_2),...,\varepsilon_{PL_p}(t_{T})]$ with size $N \times \sum{L_jT}$, which is the noise from Eq. (\ref{equ:Y-likelihood}).
    \item $\bG:$ Coefficient of noise (from Eq. (\ref{equ:matrix-bilinear-trans})) with size $\sum{L_jK_j} \times \sum{L_jK_j}$.
    \item $\bH:$ Coefficient of noise (from Eq. (\ref{equ:Y-likelihood})) with size $\sum{L_jK_j} \times \sum{L_jT}$.
    \item $\Tilde{\bG}:$ Posterior coefficient of noise (from Eq. (\ref{equ:matrix-bilinear-trans})) with size $\sum{L_jK_j} \times \sum{L_jK_j}$. 
    \item $\Tilde{\bH}:$ Posterior coefficient of noise (from Eq. (\ref{equ:Y-likelihood})) with size $\sum{L_jK_j} \times \sum{L_jT}$. 
    \item $\hat{\bG}:$ Coefficient of noise after backward smoothing (from Eq. (\ref{equ:matrix-bilinear-trans})), with size $\sum{L_jK_j} \times \sum{L_jK_j}$. \item $\hat{\bH}:$ Coefficient of noise after backward smoothing (from Eq. (\ref{equ:Y-likelihood})) with size $\sum{L_jK_j} \times \sum{L_jT}$.
\end{itemize}

Then $\bX$, $\tilde{\bX}$, $\hat{\bX}$ have following representation:
\begin{align*}
    \bx^{(n)} &= \bu^{(n)} + \bG \bxi^{(n)} + \bH \bvarepsilon^{(n)} \\
    \tilde{\bx}^{(n)} &= \tilde{\bu}^{(n)} + \tilde{\bG} \bxi^{(n)} + \tilde{\bH} \bvarepsilon^{(n)} \\
    \hat{\bx}^{(n)} &= \hbu^{(n)} + \hat{\bG} \bxi^{(n)} + \hat{\bH} \bvarepsilon^{(n)} \\
\end{align*}
where $\hbu$, $\tilde{\bu}$ and $\hbu$ represent the mean of $\bx$, $\tilde{\bx}$ and $\hat{\bx}$.


Update for prior:
\begin{align*}
    \bx^{(n)}_{j} &= \tilde{\bx}^{(n)}\bC_j + \bvarepsilon^{(n)}_j \\
    &= \sum_{k\in pa_j} \bC_{kj}^T \tilde{\bu}_k + \sum_{k\in pa_j} \bC_{kj}^T \Tilde{\bG}_k \bxi^{(n)} \\
    & \quad + \sum_{k\in pa_j} \bC_{kj}^T \Tilde{\bH}_k \bvarepsilon^{(n)} + \bvarepsilon^{(n)}_j
\end{align*}
Therefore, $\bx^{(n)}_{j} \sim \mathcal{N}(\bu_j^{(n)}, \bSigma_j)$, where:
\begin{align*}
    \bu_j^{(n)} &=  \sum_{k\in pa_j} \bC_{kj}^T \tilde{\bu}_k^{(n)} \\ 
    \bG_j &= \sum_{k\in pa_j} \bC_{kj}^T \Tilde{\bG}_k + \bI_\bG(j)\\
    \bH_j &= \sum_{k\in pa_j} \bC_{kj}^T \Tilde{\bH}_k \\
    \bSigma_j &= \omega_0^2 \bG_j \bG_j^T + \bH_j {\rm diag}(\br) \bH_j^T \\
\end{align*}
where $\bI_\bG(j)$ is a $\sum{L_jK_j} \times \sum{L_jK_j}$ matrix with the identity matrix in the submatrix corresponding to node $j$, $\bI_\bG(j)_{jj} = \bI_{L_j K_j \times L_j K_j}$.

Update for posterior:
We estimated the posterior distribution of $\bx$ in $n$-th sample,
\begin{align*}
     \bY^{(n)}_{jl} &= \bB_j\bx^{(n)}_{jl} + \bvarepsilon_{jl}^{(n)} \\
    &=  \bB_j(\bu^{(n)}_{jl} + \bG_{jl} \bxi^{(n)} + \bH_{jl} \bvarepsilon^{(n)}) + \bvarepsilon_{jl}^{(n)}
\end{align*}
Therefore, $\bY^{(n)}_{jl} \sim \mathcal{N}(\bB_j\bu^{(n)}_{jl}, \bB_j \bSigma_{jl} \bB_j^T + r_{jl}^2 \bI_T)$, where:
And we have:
\begin{equation*}
\begin{split}
    \left( 
    \begin{array}{c}
        \bx^{(n)}_{jl} \\
        \bY^{(n)}_{jl}
    \end{array}
    \right)
    \sim 
    \mathcal{N} 
    \left( \left(
    \begin{array}{c}
        \hbu_{jl}^{(n)} \\
        \bB_j\bu_{jl}^{(n)}
    \end{array}
    \right) 
    ,
    \left(
    \begin{array}{cc}
        \bSigma_{jl} & \bSigma_{jl} \bB_j^T \\
        \bB_j \bSigma_{jl}  & \bB_j \bSigma_{jl} \bB_j^T + r_{jl}^2 \bI_T
    \end{array}
    \right)
    \right)
\end{split}
\end{equation*}

The posterior $ \bx_{jl} | \bY_{jl} \sim \mathcal{N}(\tilde{\bu}_{jl}, \Tilde{\bSigma}_{jl})$, where
\begin{align*}
    \tilde{\bu}_{jl}^{(n)} &= \bu_{jl}^{(n)} + \bSigma_{jl} \bB_j^T (\bB_j \bSigma_{jl} \bB_j^T + r_{jl}^2 \bI_T) ^{-1} (\mathbf{\bY}_{jl} - \bB_j \bu_{jl}^{(n)}) \\ 
    \Tilde{\bG}_{jl} &=  \bG_{jl} - \bSigma_{jl} \bB_j^T (\bB_j \bSigma_{jl} \bB_j^T + r_{jl}^2 \bI_T)^{-1} \bB \bG_{jl}  \\
    \Tilde{\bH}_{jl} &=  \bH_{jl} - \bSigma_{jl} \bB_j^T (\bB_j \bSigma_{jl} \bB_j^T + r_{jl}^2 \bI_T) ^{-1} (\bB \bH_{jl} +  \bI_{\bH}(j,l)) 
\end{align*}
where $\bI_{\bH}(j,l)$ is a $\sum L_jT \times \sum L_j T$ matrix with the identity matrix $\bI_T$ in the submatrix corresponding to the $l$-th function in node $j$, $\bI_{\bH}(j,l)_{jl,jl} = \bI_T$.
\subsubsection{Backward smoothing}

For $k$ and the descendants $j$,  we derive the covariance of nodes $j, k$:
\begin{equation*}
    \tilde{\bSigma}_{j,k} = \omega_0^2 \Tilde{\bG}_j \Tilde{\bG}_k^T + \Tilde{\bH}_j D(\br) \Tilde{\bH}_k^T 
\end{equation*}
\begin{equation}
\small
\begin{split}
    \left( 
    \begin{array}{c}
        \Tilde{\bx}_{k}^{(n)} \\
        \Tilde{\bx}_{de(k)}^{(n)}
    \end{array}
    \right) 
    \sim 
    \mathcal{N} 
    \left( \left(
    \begin{array}{c}
        \tilde{\bu}_{k}^{(n)} \\
        \tilde{\bu}_{de(k)}^{(n)}
    \end{array}
    \right) 
    ,
    \left(
    \begin{array}{cc}
        \tilde{\bSigma}_{k} &  \tilde{\bSigma}_{k,de(k)} \\
        \tilde{\bSigma}_{k,de(k)}^T & \tilde{\bSigma}_{de(k)}
    \end{array}
    \right)
    \right)
\end{split}
\end{equation}

Derive $p(\tilde{\bx}_k|\tilde{\bx}_{de(k)}, \bY)$:  
\begin{align*}
    \hbu_{k}^{(n)} &= \tilde{\bu}_k^{(n)} + \tilde{\bSigma}_{k,de(k)} \bSigma_{de(k)}^{-1}(\hbu_{de(k)}^{(n)} - \tilde{\bu}_{de(k)}^{(n)}) \\
    \hat{\bG}_k &= \Tilde{\bG}_k - \tilde{\bSigma}_{k,de(k)} \bSigma_{de(k)}^{-1} (\tilde{\bG}_{de(k)} - \hat{\bG}_{de(k)}) \\
    \hat{\bH}_k &= \Tilde{\bH}_k - \tilde{\bSigma}_{k,de(k)} \bSigma_{de(k)}^{-1} (\tilde{\bH}_{de(k)} - \hat{\bH}_{de(k)})
\end{align*}

Finally, posterior mean of $\bx$ is $\hbu$ and the posterior variance is:
\begin{equation*}
    \hat{\bSigma} =  \omega_0^2 \hat{\bG} \hat{\bG}^T + \hat{\bH} {\rm diag}(\br) \hat{\bH}^T 
\end{equation*}

\end{document}